\documentclass[%
 reprint,
 amsmath,amssymb,
 aps,
 nofootinbib
]{revtex4-1}
\usepackage[dvipdfmx]{graphicx}
\usepackage[dvipdfmx]{xcolor}
\usepackage{dcolumn}
\usepackage{bm}
\usepackage{braket}
\usepackage{amsmath}
\usepackage{amssymb}
\usepackage{amsthm}
\usepackage{enumerate}
\usepackage{booktabs}
\usepackage{txfonts}
\usepackage[update,prepend]{epstopdf}
\usepackage[english]{babel}
\usepackage{listings}
\usepackage{qcircuit}
\usepackage[subrefformat=parens]{subcaption}
\usepackage{caption}
\usepackage{multirow}
\usepackage{eurosym}
\usepackage{algorithm}
\usepackage{algorithmic}
\usepackage{nccmath}
\usepackage{mathtools}
\newtheorem{theorem}{Theorem}
\newtheorem{lemma}{Lemma}
\newtheorem{definition}{Definition}

\allowdisplaybreaks[1]

\makeatletter
\def\hlinewd#1{%
	\noalign{\ifnum0=`}\fi\hrule \@height #1 %
	\futurelet\reserved@a\@xhline}
\makeatother

\makeatletter
\renewcommand{\ALG@name}{Procedure}
\makeatother

\begin{document}


\title{Quantum algorithms for numerical differentiation of expected values with respect to parameters}

\author{Koichi Miyamoto}
\email{koichi.miyamoto@qiqb.osaka-u.ac.jp}
\affiliation{Center for Quantum Information and Quantum Biology, Osaka University \\ 1-3 Machikaneyama, Toyonaka, Osaka, 560-8531, Japan}


\date{\today}

\begin{abstract}
	
The quantum algorithms for Monte Carlo integration (QMCI), which are based on quantum amplitude estimation (QAE), speed up expected value calculation compared with classical counterparts, and have been widely investigated along with their applications to industrial problems such as financial derivative pricing.
In this paper, we consider an expected value of a function of a stochastic variable and a real-valued parameter, and how to calculate derivatives of the expectation with respect to the parameter.
This problem is related to calculating sensitivities of financial derivatives, and so of industrial importance.
Based on QMCI and the general-order central difference formula for numerical differentiation, we propose two quantum methods for this problem, and evaluate their complexities.
The first one, which we call the naive iteration method, simply calculates the formula by iterative computations and additions of the terms in it, and then estimates its expected value by QAE.
The second one, which we name the sum-in-QAE method, performs the summation of the terms at the same time as the sum over the possible values of the stochastic variable in a single QAE.
We see that, depending on the smoothness of the function and the number of qubits available, either of two methods is better than the other.
In particular, when the function is nonsmooth or we want to save the qubit number, the sum-in-QAE method can be advantageous.

\end{abstract}

\maketitle


\section{\label{sec:intro}Introduction}

Following the recent rapid development of quantum computing technologies, many researchers are now investigating applications of quantum algorithms to practical problems in industries.
The quantum algorithms for Monte Carlo integration (MCI), which we hereafter abbreviate as QMCI, are representative examples \cite{Montanaro,Suzuki,Herbert}.
Based on quantum amplitude estimation (QAE) \cite{Brassard,Suzuki,Aaronson,Grinko,Nakaji,Brown,Kerenidis,Giurgica-Tiron,Tanaka,Uno,Wang}, they output an estimate on an expected value $E[F(S)]$ of a function $F$ of a stochastic variable $S$ with an error up to $\epsilon$, making $O\left(\frac{1}{\epsilon}\right)$ calls to some oracles such as $O_F$, which calculates $F$, and $O_S$, which generate a quantum state corresponding to the probability distribution of $S$ (see Section \ref{sec:QAE} for the detail).
This is often referred to as the {\it quadratic speedup} compared with classical counterparts, which have a query complexity of $O(\epsilon^{-2})$.

One of application targets of QMCI is finance, especially {\it financial derivative pricing} (for readers unfamiliar with this, we refer to \cite{Hull} as a textbook).
A financial derivative is a contract between two parties, in which amounts (payoffs) determined by some underlying asset prices are paid and received.
Roughly speaking, a financial derivative price is given by an expected payoff under a given mathematical finance model for time evolution of underlying asset prices, and calculated typically by MCI in which asset price evolution is simulated.
Large banks, which have large portfolios of financial derivatives, take a lot of time and cost for pricing calculation in daily business, and therefore QMCI is expected to provide a large benefit to financial industry through quantum speedup.

In this paper, we consider how to calculate derivatives\footnote{In this paper, when we simply say a {\it derivative}, it refers to a mathematical terms, that is, a derivative of a function. On the other hand, when we refer to a derivative as a financial product, we use a {\it financial derivative}.} of an expected value with a parameter.
Although there are some previous researches \cite{Jordan,Gilyen,Cornelissen} on quantum algorithms to calculate derivatives of a given function, this paper is the first one focusing on quantum algorithms for derivatives of an expected value, as far as the author knows.
For a function $F(S,x)$ of a stochastic variable $S$ and a real number $x$, its expected value $V(x):=E[F(S,x)]$ with respect to the randomness of $S$ for fixed $x$ can be viewed as a function of $x$.
We often calculate the derivatives of $V(x)$ such as $V^\prime(x)$ and $V^{\prime\prime}(x)$.
For example, for a financial derivative, a bank calculates not only its price but also the derivatives of the price with respect to input variables such as the present underlying asset price and model parameters.
These are called {\it sensitivities} or {\it Greeks}, and have crucial roles for risk management in financial derivative business \cite{Hull}.

In many practical problems, when $V(x)$ does not have a closed formula, neither do its derivatives.
Then, a naive way to calculate a derivative is some {\it finite difference} method such as {\it central difference}:
\begin{equation}
	V'(x)\approx\frac{V(x+h)-V(x-h)}{2h} \label{eq:centIntro}
\end{equation}
where $h$ is some positive real number.
This is the lowest order approximation formula for the first derivative, and there are the higher order approximation formulas for higher order derivatives, which have more terms and use the values of $V$ at a larger number of grid points on the $x$ axis (see Section \ref{sec:numDiff} for the detail). 
Basically, taking smaller $h$ leads to the higher accuracy.
However, when $V$ is calculated by some numerical method such as MCI, there is a subtlety.
Namely, it is not appropriate to calculate $V(x+h)$ and $V(x-h)$ individually and then use (\ref{eq:centIntro}).
This is because the result of a numerical calculation accompanies an error.
Suppose that we have erroneous estimates $\tilde{V}(x+h)=V(x+h)+\epsilon_1$ and $\tilde{V}(x-h)=V(x-h)+\epsilon_2$, where $\epsilon_1$ and $\epsilon_2$ are numerical errors and their absolute values are bounded by $\epsilon$ with high probability.
Plugging these into (\ref{eq:centIntro}) yields
\begin{equation}
	\frac{V(x+h)-V(x-h)}{2h} + \frac{\epsilon_1-\epsilon_2}{2h}.
\end{equation}
Here, the second term is of $O\left(\frac{\epsilon}{h}\right)$.
Divided by a small number $h$, this can be large even if the error level $\epsilon$ is suppressed well.
In other word, to accurately calculate a derivative, we have to suppress an error tremendously so that the dividing factor $h$ is compensated, which results in the large complexity.

In classical MCI, there exist some solutions to this issue \cite{Glasserman}.
For example, in a common way where we use a pseudo-random number sequence on behalf of random variables, using a same seed for $V(x+h)$ and $V(x-h)$ alleviates this difficulty, since $\epsilon_1$ and $\epsilon_2$ become close and cancel each other. 

In the case of QMCI, how can we calculate derivatives?
Contrary to classical MCI, it is difficult to cancel errors in different runs of QMCI, since they are independent by quantum nature. 
Instead, we consider the following way:
\begin{equation}
	V'(x)\approx E\left[\frac{F(S,x+h)-F(S,x-h)}{2h}\right]. \label{eq:solIntro}
\end{equation}
That is, we calculate not the difference quotient of $V$ but the expected value of the difference quotient of $F$ with respect to $x$.
There is only the error from one QMCI for (\ref{eq:solIntro}), and therefore we do not have to concern the aforementioned issue.

However, there are other subtleties in this approach.
First, it is possible that $V$ is smooth but $F$ is nonsmooth or even discontinuous.
This often happens in financial derivative pricing, because nonsmooth payoff functions are ubiquitous in practice, as we will see in Section \ref{sec:problem}.
In such a case, for some $(S,x)$, the smaller $h$ we take, the larger the difference quotient of $F$ becomes.
As explained in Section \ref{sec:QAE}, in QMCI, we normalize the integrand so that it is in the interval $[0,1]$ and encode it into the amplitude of a qubit.
Since the complexity of QMCI becomes larger for the larger normalization factor, the nonsmoothness of $F$ can lead to the large complexity.
Second, even if $F$ is smooth, taking small $h$ for accuracy can cause the issue on the {\it qubit number} as follows.
The smaller $h$ is, the closer $F(S,x+h)$ and $F(S,x-h)$ are.
Whether it is classical or quantum, a computer can perform only the finite precision computation, and, in order to avoid the cancellation of significant digits between $F(S,x+h)$ and $F(S,x-h)$, we have to calculate them with sufficiently high precision.
This means that, for smaller $h$, we use the larger number of qubits to compute $F$.
Although the similar issue on memory size exist also in classical MCI, the severity is higher in QMCI, since the large qubit overhead for error correction might largely limit the number of logical qubits available even in the future \cite{Fowler}.
In particular, for problems like financial derivative pricing, where $F$ is computed through the complicated procedure, this qubit issue is more serious.

Taking into account these points, this paper proposes two quantum methods for numerical differentiation of $V$, and evaluates their complexities in terms of the numbers of queries to the oracles such as $O_F$ and $O_S$, focusing on their dependencies on the error tolerance $\epsilon$.
The first method, which we call the {\it naive iteration method}, simply calls $O_F$ iteratively to calculate the finite difference formula for $F$ term-by-term, and estimates its expected value by QAE.
The second method, which we name the {\it sum-in-QAE method}, utilizes the quantum parallelism more deeply and is more nontrivial.
That is, in this method, we perform the summation of terms in the finite difference formula {\it at the same time} as the sum over the possible values of $S$ in one QMCI.
As we will see below, when $F$ is smooth and we can use so many qubits that we can take sufficiently small $h$, the naive iteration method is better in the aspect of query complexity, since we can use the lowest order difference formula.
On the other hand, when $F$ is nonsmooth, the sum-in-QAE method with the high order formula and large $h$ is appropriate.
Besides, even when $F$ is smooth, if we save the qubit number as much as possible, the sum-in-QAE method can be more advantageous, depending on the parameter measuring the smoothness of $F$, which is introduced in Section \ref{sec:problem}.

The rest of this paper is organized as follows.
In Section \ref{sec:prel} is a preliminary one, which explains the notation we use in this paper, and gives brief reviews on numerical differentiation and QMCI.
Section \ref{sec:qAlgo} is the main part of this paper, where we present the naive iteration method and the sum-in-QAE method, and evaluate and compare their complexities in the various situations on smoothness of $F$ and qubit capacity.
Section \ref{sec:sum} summarizes this paper.

\section{\label{sec:prel}Preliminaries}

\subsection{Notation}

$\mathbb{N}$ denotes the set of all positive integers, and $\mathbb{N}_0:=\mathbb{N}\cup\{0\}$ is the set of all non-negative integers.
For every $x\in \mathbb{R}$, $\mathbb{N}_{\ge x}:=\{i\in\mathbb{N} \ | \ i\ge x\}$ is the set of all positive integers not less than $x$.
For every integer pair $(m,n)$ satisfying $m\le n$, we define $[m:n]:=\{i\in\mathbb{Z} \ | \ m\le i \le n\}$, where $\mathbb{Z}$ is the set of all integers. 

$\mathbb{R}_+$ is the set of all positive real numbers, and $\mathbb{R}_{\ge 0}:=\mathbb{R}_+\cup\{0\}$ is the set of all non-negative real numbers.

We denote the set of all $k$-combinations from a finite set $E$ as $\mathcal{P}_k(E)$, where $k\in[|E|]$.

For given $x\in\mathbb{R}$ and $\epsilon\in\mathbb{R}_+$, we call any $y\in\mathbb{R}$ satisfying $|x-y|\le\epsilon$ a $\epsilon$-approximation of $x$.

For every $x\in\mathbb{R}$, $\ket{x}$ denotes a computational basis state on some quantum register, in which the bit string on the register corresponds to a finite precision binary representation of $x$.


\subsection{\label{sec:numDiff}Numerical differentiation}

Now, let us briefly review numerical differentiation.
It is the method to approximately calculate a derivative of a given real-valued function $f$ on some interval on $\mathbb{R}$ in the case that we can calculate $f$ but not its derivatives directly.
Based on the definition that $f^\prime(x)=\lim_{h\rightarrow 0}\frac{f(x+h)-f(x)}{h}$, we can approximate $f^\prime(x)$ as
\begin{equation}
	f^\prime(x) = \frac{f(x+h)-f(x)}{h} + O(h),
\end{equation}
using a sufficiently small positive real number $h$.
This type of approximation is called the forward difference method.
The residual error term scales as $O(h)$ in this scheme.
However, there is the {\it central difference method}, in which the error scales as $O(h^2)$, and it is therefore used more often:
\begin{equation}
	f^\prime(x) = \frac{f(x+h)-f(x-h)}{2h} + O(h^2). \label{eq:cent1st}
\end{equation} 
Also in this paper, we consider this method.
In fact, (\ref{eq:cent1st}) is the lowest-order approximation in this way.
There are the higher order approximations for higher order derivatives.
The general order formula and its residual term were investigated in \cite{Li}.
Here, we present them as the following theorem, which is same as Corollary 2.1 in \cite{Li} except some slight changes. 
\begin{theorem}[Corollary 2.1 in \cite{Li}, modified]
	Let $m,n$ be positive integers such that $m\le 2n$.
	Let $f$ be a function such that $f\in C^{2n+1}(\mathbb{R})$.
	Then, for any $x\in\mathbb{R}$ and $h\in\mathbb{R}_+$,
	\begin{eqnarray}
		f^{(m)}(x)&=&\mathcal{D}_{n,m,h}[f](x)+R_f(x,n,m,h)h^{2n-m+1} \nonumber \\
		\mathcal{D}_{n,m,h}[f](x)&:=&\frac{1}{h^m}\sum_{j=-n}^{n} d^{(m)}_{n,j}f(x_j)\nonumber \\
		d^{(m)}_{n,j} &:=& 
		\begin{dcases}
			\frac{(-1)^{m-j}m!a^{(m)}_{n,j}}{\left(n+j\right)!\left(n-j\right)!} &; \ {\rm for} \ j\in[-n:n]\setminus\{0\} \\
			-\sum_{j\in[-n:n]\setminus\{0\}} d^{(m)}_{n,h,j}&; \ {\rm for} \ j=0
		\end{dcases}
		\nonumber \\
		R_f(x,n,m,h)&:=& \frac{(-1)^{m+1}m!}{(2n+1)!}\sum_{j\in[-n:n]\setminus\{0\}}\frac{(-1)^jf^{(2n+1)}(\xi_j)j^{2n+1}a^{(m)}_{n,j}}{\left(n+j\right)!\left(n-j\right)!} \nonumber\\
		a^{(m)}_{n,j} &:=& \sum_{\substack{\{l_1,...,l_{2n-m}\}\in \qquad\quad\\ \mathcal{P}_{2n-m}([-n:n]\setminus\{0,j\})}} \prod_{i=1}^{2n-m}l_i \quad {\rm for} \ j\in[-n:n]\setminus\{0\} \nonumber\\
		&&\label{eq:centGen}
	\end{eqnarray}
	holds, where, for every $j\in\left\{-n,-n+1,...,n\right\}$, $x_j:=x+hj$ and $\xi_j$ is some real number depending on $x$ and $x_j$.
	\label{th:numDiff}
\end{theorem}

\noindent This theorem states that the central difference method using the values of $f$ at $2n+1$ points with interval $h$ outputs an approximation of $f^{(m)}$ with an error of $O(h^{2n-m+1})$.

Let us comment on a virtue of central difference in the case of odd $m$, which includes the formula (\ref{eq:cent1st}) for the first derivative.
In this case, $d^{(m)}_{n,0}=0$ holds, as mentioned in \cite{Li}.
In particular, for $n=\left\lceil\frac{m}{2}\right\rceil=\frac{m+1}{2}$, the minimum value of $n$, $d^{(m)}_{\frac{m+1}{2},0}=0$ holds.
That is, although (\ref{eq:centGen}) seemingly requires evaluating $f$ at $2n+1=m+2$ points, we actually need to evaluate $f$ only at $m+1$ points, which is equal to the minimum number of points to calculate $f^{(m)}(x)$ by finite difference formulas.
Nevertheless, the error is $O(h^{2})$.
This is contrasted with other methods such as forward difference method, which uses the $m+1$ values $f(x), f(x+h),...,f(x+mh)$ and gives an estimate of $f^{(m)}(x)$ with an error of $O(h)$.

Let us comment also on the domain of $f$.
Originally, \cite{Li} presented Corollary 2.1 for a function $f$ on some interval on $\mathbb{R}$.
However, considering a bounded domain may make the discussion cumbersome, e.g., we must care whether the points in the numerical differentiation formula are within the domain.
Therefore, for simplicity, we simply consider the case where $f$ is defined and sufficiently smooth on $\mathbb{R}$ in this paper.
Theorem \ref{th:numDiff} has been modified in such a way.
We expect that extending the discussion to the case where the domain is bounded is possible with the essential parts not affected.

\subsection{Quantum amplitude estimation and its application to Monte Carlo integration \label{sec:QAE}}

We now present a brief explanation on QAE \cite{Brassard,Suzuki,Aaronson,Grinko,Nakaji,Brown,Kerenidis,Giurgica-Tiron,Tanaka,Uno,Wang}.
Suppose that we want to solve the following problem: given an oracle $A$ on a system of a quantum register $R_1$ and a qubit $R_2$ such that
\begin{equation}
	A\ket{0}\ket{0}=\sqrt{a}\ket{\psi_1}\ket{1}+\sqrt{1-a}\ket{\psi_0}\ket{0}
\end{equation}
with some $a\in(0,1)$ and some quantum states $\ket{\psi_0}$ and $\ket{\psi_1}$ on $R_1$, estimate $a$, the probability that we obtain $1$ on $R_2$ in $A\ket{0}\ket{0}$, up to an error at most $\epsilon\in\mathbb{R}_+$.
QAE is a quantum algorithm for this.
Making $O\left(\frac{1}{\epsilon}\right)$ calls to $A$ and $A^{\dagger}$, QAE outputs $\epsilon$-approximation of $a$ with a probability higher than a given value (say, 0.99).

There are some applications of QAE, and Monte Carlo integration is one of them \cite{Montanaro,Suzuki,Herbert}.
Suppose that we want to calculate
\begin{equation}
	E[F(S)]=\sum_{s\in\mathcal{S}} p_sF(s) \label{eq:expVal}
\end{equation}
with an error up to $\epsilon\in\mathbb{R}_+$.
Here, $S$ is a stochastic variable which takes an element $s$ in some finite set $\mathcal{S}$ with a probability $p_s\in[0,1]$, and $F$ is a bounded real-valued function on $\mathcal{S}$, which satisfies $|F(s)|\le C$ for any $s\in\mathcal{S}$ with some $C\in\mathbb{R}_+$. 
We assume the availability of the following oracles.
The first one is $O_S$, which generates a quantum state corresponding to the distribution of $S$:
\begin{equation}
	O_S\ket{0}=\sum_{s\in\mathcal{S}} \sqrt{p_s} \ket{s}. \label{eq:OraS}
\end{equation}
Here, we assume that elements in $\mathcal{S}$ are associated with mutually different real numbers, and $\ket{s}$ denotes a computational basis state corresponding to the real number for $s\in\mathcal{S}$.
Note that, in measuring this state, we obtain $s$ with a probability $p_s$.
The second one is $O_{F}$, which calculate $F(s,x)$ for every $(s,x)\in\mathcal{S}\times\mathbb{R}$:
\begin{equation}
	O_{F}\ket{s}\ket{0}=\ket{s}\ket{F(s)}. \label{eq:OraFGen}
\end{equation}
Using these oracles, we can perform the following operation on a three-register system, where the last one is an ancilla qubit:
\begin{eqnarray}
	&&\ket{0}\ket{0}\ket{0} \nonumber \\
	&\rightarrow& \sum_{s\in\mathcal{S}} \sqrt{p_s} \ket{s}\ket{0}\ket{0} \nonumber \\
	&\rightarrow& \sum_{s\in\mathcal{S}} \sqrt{p_s} \ket{s}\ket{F(s)}\ket{0} \nonumber \\
	&\rightarrow& \sum_{s\in\mathcal{S}} \sqrt{p_s} \ket{s}\ket{F(s)}\left(\sqrt{\frac{1}{2}+\frac{F(s)}{2C}}\ket{1}+\sqrt{\frac{1}{2}-\frac{F(s)}{2C}}\ket{0}\right). \nonumber \\
	&& \label{eq:Q}
\end{eqnarray}
Here, we used $O_S$ at the first arrow, $O_{F}$ at the second arrow, and some arithmetic circuits \cite{Vedral,Draper,Cuccaro,Takahashi,Draper2,Takahashi2,AlvarezSanchez,Takahashi3,Thapliyal,Thapliyal2,Jayashree,MunozCoreas,Khosropour,Dibbo,Thapliyal3,MunozCoreas2} and a controlled rotation gate at the third arrow.
Then, the probability to obtain 1 on the ancilla in the last state is
\begin{equation}
	P=\frac{1}{2}+\frac{1}{2C}\sum_{s\in\mathcal{S}} p_sF(s),
\end{equation}
which means that
\begin{equation}
	E[F(S)]=C(2P-1). \label{eq:EP}
\end{equation}
Therefore, we obtain an approximation of $E[F(S)]$ by estimating $P$ using QAE and calculating (\ref{eq:EP}).
If we want to estimate $E[F(S)]$ with an error up to $\epsilon$, it is sufficient to estimate $P$ with an error of $O\left(\frac{\epsilon}{C}\right)$.
This means that the oracle $Q$, which corresponds to the operation (\ref{eq:Q}), is called $O\left(\frac{C}{\epsilon}\right)$ times, and so are $O_S$ and $O_{F}$, since $Q$ contains one each of them. 

Although we assume that the domain $\mathcal{S}$ of $S$ is finite, we often consider unbounded and/or continuous stochastic variables, such as a normal random variable which can take any real number.
Such a case can be boiled down to the above setup by a discrete approximation.
That is, we can set lower and upper bounds for $S$ and grid points between the bounds, and approximate $S$ as a discrete stochastic variable taking any of the grid points.
Actually, there are quantum algorithms for generating a quantum state like (\ref{eq:OraS}) which corresponds to a discretely approximated stochastic variable \cite{Grover,Kaneko}.
Also note that, in quantum computation, the error from discrete approximation can be exponentially suppressed, which means that we can deal with $2^n$ grid points using $n$ qubits.
Hereafter, we consider that (\ref{eq:expVal}) covers the cases where $\mathcal{S}$ is not finite through the discrete approximation, and neglect such an approximation error.

\section{\label{sec:qAlgo}Quantum algorithm for numerical differentiation of expected values}

\subsection{\label{sec:problem} Problem}

Hereafter, we consider the following problem.
As above, let $S$ be a stochastic variable taking an element $s$ in some finite set $\mathcal{S}$ with a probability $p_s\in[0,1]$.
Suppose that we are given a map $F:\mathcal{S}\times \mathbb{R} \rightarrow \mathbb{R}$.
We then define the expected value of $F$ as
\begin{equation}
	V(x):=\sum_{s\in\mathcal{S}} p_sF(s,x) \label{eq:V}
\end{equation}
for every $x\in\mathbb{R}$, and regard this as a real-valued function of a real number $x$, which we call a {\it parameter}.
Now, for given $x$, we want to calculate not only $V(x)$ but also $V^{(m)}(x)$, the $m$-th derivative of $V$ at $x$, where $m$ is a given positive integer.

In fact, this covers pricing and sensitivity calculation for financial derivatives.
In these calculations, paths of time evolution of the underlying asset price are generated by some stochastic variables, and the financial derivative price is given as the expectation value of the payoff determined by the asset price.
As a concrete example, let us consider the following problem.
Consider a financial derivative written on the underlying asset whose price $P_t$ at time $t$ obeys the Black-Scholes model with a volatility $\sigma$ and a risk-free rate $r$.
At its maturity $T$, the payoff $f_{\rm pay}(P_T)$ arises, where $f_{\rm pay}:\mathbb{R}_+\rightarrow\mathbb{R}$ is some function, and $P_T$, the asset price at $T$, is given by
\begin{equation}
	P_T = P_0 \exp\left(\sigma \sqrt{T} S + \left(r-\frac{1}{2}\sigma^2\right)T\right)
\end{equation}
with $P_0\in\mathbb{R}_+$, the asset price at the present $t=0$, and a standard normal random variable $S$.
Then, the present price of this contract is calculated as
\begin{eqnarray}
	&&V(P_0,\sigma,r) = \nonumber \\
	&& \qquad \int_{-\infty}^{+\infty} ds \phi_{\rm SN}(s) f_{\rm pay}\left(P_0 \exp\left(\sigma \sqrt{T} s + \left(r-\frac{1}{2}\sigma^2\right)T\right)\right), \nonumber \\
	&& \label{eq:BSPrice}
\end{eqnarray}
where $\phi_{\rm SN}$ is the density function of the standard normal distribution (see \cite{Hull} for the details).
We can make the price (\ref{eq:BSPrice}) correspond to (\ref{eq:V}), viewing any of $P_0$, $\sigma$, and $r$ as $x$, and discretely approximating $S$.
For example, if $x$ is $P_0$, we can consider that $F(s,x)=f_{\rm pay}\left(x \exp\left(\sigma \sqrt{T} s + \left(r-\frac{1}{2}\sigma^2\right)T\right)\right)$.
This example is one of the simplest problems in practical pricing tasks, and banks often deal with more complicated contracts and use more advanced models.
In such a case, the financial derivative price is often not expressed simply as (\ref{eq:BSPrice}), but calculated by some numerical method such as MCI, where many stochastic variables and parameters are involved.

For a bank, it is important to calculate not only a financial derivative price but also its derivatives.
These are called {\it sensitivities} or {\it Greeks}, and have crucial roles in risk management.
For example, in the above example, the following are representative ones: the first derivatives with respect to $P_0$, $\sigma$ and $r$, which are called the {\it delta}, the {\it vega} and the {\it rho}, respectively, and the second derivatives with respect to $P_0$, which is called the {\it gamma}.

Then, how can we calculate $V^{(m)}(x)$, when $V$ does not have the closed formula and neither do its derivatives?
One way is the central difference method described in Sec. \ref{sec:numDiff}.
That is, choosing $n\in\mathbb{N}_{\ge\frac{m}{2}}$ and $h\in\mathbb{R}_+$, and assuming that $V\in C^{2n+1}(\mathbb{R})$, we can approximate $V^{(m)}(x)$ by
\begin{equation}
	V^{(m)}(x) \approx \mathcal{D}_{n,m,h}[V](x) = \sum_{s\in\mathcal{S}} p_s \mathcal{D}_{n,m,h}[F(s,\cdot)](x), \label{eq:Vm}
\end{equation}
where
\begin{equation}
	\mathcal{D}_{n,m,h}[F(s,\cdot)](x) = \frac{1}{h^m}\sum_{j=-n}^{n} d^{(m)}_{n,j} F(s,x+jh). \label{eq:DF}
\end{equation}
Now, we have the following questions:
\begin{itemize}
	\item can we construct a quantum algorithm to calculate the above numerical differentiation?
	\item what is the best setting, e.g. $n$ and $h$, to reduce complexity keeping the desired accuracy?
\end{itemize}

For such a quantum algorithm, it is plausible to assume the availability of the following oracles.
The first one is $O_S$ mentioned above, which performs the operation (\ref{eq:OraS}).
The second one is $O_{F}$, which is similar to (\ref{eq:OraFGen}), but now calculate $F(s,x)$ for every $(s,x)\in\mathcal{S}\times\mathbb{R}$:
\begin{equation}
	O_{F}\ket{s}\ket{x}\ket{0}=\ket{s}\ket{x}\ket{F(s,x)}. \label{eq:OraF}
\end{equation}
In the context of financial derivative pricing, $O_{F}$ generates an asset price path and computes a payoff.
Hereafter, we consider that $O_{F}$ consumes a longer time and more qubits than $O_S$.
This is the case in some practical problems such as financial derivative pricing under an advanced model, which is the very target of quantum speedup, since the asset price evolution is based on some complicated stochastic differential equation.
In particular, in the QMCI method proposed in \cite{Miyamoto}, in which, for the sake of qubit reduction, we calculate the integrand using pseudo-random numbers (PRNs) sequentially generated on a single register, $O_S$ corresponds to generation of an equiprobable superposition of integer indexes which specify the start point of a PRN sequence and is implemented just by a set of Hadamard gates, whereas $O_F$ contains complicated calculations such as asset price evolution and PRN generation.
We hence focus on $O_{F}$ in the discussion on query complexity and qubit number in the quantum algorithms proposed later.
Practically, it is more plausible to consider the following oracle $O_{F,\epsilon}$ rather than $O_F$.
$O_{F,\epsilon}$ performs the operation
\begin{equation} 
	O_{F,\epsilon} : \ket{s}\ket{x}\ket{0}\mapsto\ket{s}\ket{x}\ket{F_{\epsilon}(s,x)} \label{eq:OFeps}
\end{equation}
for every $(s,x)\in\mathcal{S}\times\mathbb{R}$.
Here, $\epsilon\in\mathbb{R}_+$ and $F_{\epsilon}:\mathcal{S}\times\mathbb{R}\rightarrow\mathbb{R}$ is a function such that
\begin{equation}
	\forall (s,x)\in\mathcal{S}\times\mathbb{R}, |F_{\epsilon}(s,x)-F(s,x)|\le\epsilon, \label{eq:FFepsdiff}
\end{equation}
and uses
\begin{equation}
	O\left(\log^a\left(\frac{1}{\epsilon}\right)\right) \label{eq:qubitOFeps}
\end{equation}
qubits including ancillas, where $a\in\mathbb{R}_+$.
This reflects the fact that we can perform only finite precision computation on a quantum computer and better precision requires more qubits.
For example, if we assume that calculation of $F$ can be decomposed into arithmetic circuits \cite{Vedral,Draper,Cuccaro,Takahashi,Draper2,Takahashi2,AlvarezSanchez,Takahashi3,Thapliyal,Thapliyal2,Jayashree,MunozCoreas,Khosropour,Dibbo,Thapliyal3,MunozCoreas2,Haner}, $n$-bit precision typically requires $O(n)$ qubits, which corresponds to (\ref{eq:qubitOFeps}) with $a=1$.
Of course, similar issues exist also for $O_S$ and other circuits used in QMCI, but we hereafter consider this point only for $O_F$ because of the assumption that it is most costly in terms of qubits.

For a quantitative discussion on accuracy and complexity of algorithms, we need some assumption on derivatives of $V$, which are our targets and affect the residual terms of the formula (\ref{eq:centGen}).
In this paper, following \cite{Cornelissen}, we consider functions with the following property, which are called Gevrey functions.

\begin{definition}
	Let $A,c\in\mathbb{R}_+$ and $\sigma\in\mathbb{R}$. 
	The set of functions $f:\mathbb{R}\rightarrow\mathbb{R}$ such that 
	\begin{equation}
		|f^{(k)}(x)|\le Ac^k (k!)^\sigma \label{eq:Gev}
	\end{equation}
	for any $k\in\mathbb{N}_0$ and $x\in\mathbb{R}$ is denoted by $\mathcal{G}_{A,c,\sigma}$.
	\label{def:Gevrey}
\end{definition}

\noindent Hereafter, we say that a function $f:\mathbb{R}\rightarrow\mathbb{R}$ is {\it smooth} if $f\in\mathcal{G}_{A,c,\sigma}$ with some $A,c\in\mathbb{R}_+$ and $\sigma\in\mathbb{R}$.

Now, let us make some comments on the aforementioned setup.
First, although we assumed that the domain of the second variable (the parameter) of $F$ is $\mathbb{R}$, it might be some subset $U$ of $\mathbb{R}$ in an actual problem.
For example, when we consider financial derivative pricing and regard the parameter $x$ as the initial underlying asset price, $x$ is often limited to $\mathbb{R}_+$.
Besides, unlike \cite{Cornelissen}, Definition \ref{def:Gevrey} refers to functions on $\mathbb{R}$ but not those on any subsets.
These are just because we want to avoid cumbersomeness concerning to boundaries, which is mentioned in Section \ref{sec:numDiff}.
We expect that we can extend the following discussion to the case of more general domains with no modification in essence.
We also mention that we can often set the domain to $\mathbb{R}$ by variable transformation such as $\log x$, which maps $x\in\mathbb{R}_+$ into $\mathbb{R}$.

Second, note that the condition (\ref{eq:Gev}) is slightly different from that in \cite{Cornelissen}.
That is, we introduced a constant $A$, which does not appear in \cite{Cornelissen}.
This $A$ bounds the value of the function $f$ itself, and therefore informally represents the `typical scale' of $f$.
On the other hand, $c^{-1}$ roughly represents the 'scale of variation of $f$'.
In other words, we can informally say that $f(x)$ changes by roughly $A$ when the variable $x$ changes by $c^{-1}$.
This view is helpful especially when the function $f$ and the variable $x$ have different dimensions, e.g., in financial derivative pricing, $f$ as the price is written in the unit such as \$ and \euro, whereas the parameters such as the volatility have different units.  

Third, Definition \ref{def:Gevrey} is different from \cite{Cornelissen} also in that it does not cover differentiation of multivariate functions by multiple variables, whereas \cite{Cornelissen} does.
This is because, in this paper, we focus on differentiation by a single variable for simplicity.
Note that, even for a multivariate function, when we differentiate it with respect to one of its variables, we can treat it as univariate by fixing remaining variables.
We expect that it is straightforward to extend the following discussion to cross partial derivatives using the difference formulas for them.

Finally, let us comment on smoothness of $F$.
One might think that it is more natural to put smoothness conditions on $F(s,x)$ as a function of $x$ than on $V$.
However, in practice, it is possible that $F$ is not smooth but $V$ is.
For example, in the aforementioned financial derivative pricing example, $f_{\rm pay}$ is often nonsmooth or even discontinuous, e.g. $f_{\rm pay}(P)=(P-K)^+$ for a call option and
\begin{equation}
	f_{\rm pay}(P)=
	\begin{cases}
		1 & ; \ {\rm if} \ P\ge K \\
		0 & ; \ {\rm if} \ P<K
	\end{cases} \nonumber
\end{equation}
for a digital option, where $K$ is some real number.
In such a case, $F(s,x)$ is also nonsmooth or discontinuous with respect to $x$.
Nevertheless, $V$ in (\ref{eq:BSPrice}) is smooth with respect to $P_0$, $\sigma$ and $r$, because of its definition as the integral of $\phi_{\rm SN}(s)F(s,x)$ with respect to $s$.
We may concern that, even though the original $V$ is such an integral, we now consider $V$ in (\ref{eq:V}), which is a finite sum arising from the discrete approximation, and therefore $V$ is nonsmooth if so is $F$.
On the other hand, as explained in Section \ref{sec:QAE}, we now consider that the difference between $V$ in (\ref{eq:V}) and the original $V$ is negligible, and therefore so is the difference between the numerical differentiation values of these.
Hence, in the following error analysis of the central difference formula for $V$ in (\ref{eq:V}), we presume that it has the smoothness property as the original $V$.

\subsection{\label{sec:lem} Lemmas}

Here, we present some lemmas for later use.
Lemmas \ref{lem:cj}, \ref{lem:R}, and \ref{lem:cenAppOrd} are proven in Appendices \ref{sec:PrLemR}, \ref{sec:PrLemcj} and \ref{sec:PrLemCenAppOrd}, respectively.

\begin{lemma}
	Let $m$ and $n$ be positive integers satisfying $m\le 2n$.
	Let $f\in C^{2n+1}(\mathbb{R})$ be a function satisfying the following: there exists $M\in\mathbb{R}$ such that
	\begin{equation}
		\left|f^{(2n+1)}(x)\right|\le M \label{eq:f2n1Bound}
	\end{equation}
	holds for any $x\in\mathbb{R}$.
	Then, for any $h\in\mathbb{R}_+$ and any $x\in\mathbb{R}$, $R_f(x,n,m,h)$ given as (\ref{eq:centGen}) satisfies
	\begin{equation}
		|R_f(x,n,m,h)|\le Mm \left(\frac{em}{2}\right)^{2n}. \label{eq:RnUB}
	\end{equation}
	\label{lem:R}
\end{lemma}

\begin{lemma}
	For any positive integers $m$ and $n$ satisfying $m\le 2n$, define
	\begin{equation}
		D^{(m)}_{n}:=\sum_{j=-n}^n \left|d^{(m)}_{n,j}\right|
	\end{equation}
	with $d^{(m)}_{n,j}$ in (\ref{eq:centGen}).
	Then,
	\begin{equation}
		D^{(m)}_n \le 2m\left[2\left(1+\log n\right)\right]^m \label{eq:CnUB}
	\end{equation}
	holds.
	\label{lem:cj}
\end{lemma}

\begin{lemma}
	Let $f:\mathbb{R}\rightarrow\mathbb{R}$ be in $\mathcal{G}_{A,c,\sigma}$ for some $A,c\in\mathbb{R}_+$ and $\sigma\in\mathbb{R}$. 
	Then, for any $x\in\mathbb{R}$, $\epsilon\in\mathbb{R}_+$, $m\in\mathbb{N}$, $n\in\mathbb{N}$, and $h\in\mathbb{R}_+$ satisfying
	\begin{equation}
		\epsilon^\prime\le 
		\begin{cases}
			2^{m\sigma^+-\left(\frac{m\sigma^+}{\log 2}\right)^2} &; \ {\rm if} \ m\sigma^+ \ge \log 2 \\
			2^{m\sigma^+-1} &; \ {\rm otherwise}
		\end{cases},
		\label{eq:epsCond}
	\end{equation}
	\begin{equation}
		h\le h_{\rm th}:=\frac{1}{ecm(2n+1)^{\sigma^+}}, \label{eq:hcond}
	\end{equation}
	and
	\begin{equation}
		n\ge n_{\rm th} := \left\lceil\frac{1}{2}\left[\log_2 \left(\frac{2^{m\sigma^+}}{\epsilon^\prime}\right)+\log_2\left(\log_2 \left(\frac{2^{m\sigma^+}}{\epsilon^\prime}\right)\right)-\frac{1}{2}\right]\right\rceil, \label{eq:nth}
	\end{equation}
	where
	\begin{equation}
		\epsilon^\prime := \frac{e}{2(ecm)^m}\cdot\frac{\epsilon}{A}, \label{eq:epsprime}
	\end{equation}
	the following holds:
	\begin{equation}
		\left|f^{(m)}(x)-\mathcal{D}_{n,m,h}[f](x)\right|\le \epsilon. \label{eq:diffErrUB}
	\end{equation}
	
	\label{lem:cenAppOrd}
\end{lemma}

\subsection{\label{sec:algoNAive} The naive iteration method}

Now, let us consider quantum methods to compute $V^{(m)}(x)$ by numerical differentiation.
Naively thinking, we conceive the following method, which we hereafter call the {\it naive iteration method}.
That is, taking appropriate $n\in\mathbb{N}_{\ge \frac{m}{2}}$ and $h\in\mathbb{R}_+$ and supposing that we are given an oracle $O_F$ in (\ref{eq:OraF}) (in reality, $O_{F,\epsilon}$ in (\ref{eq:OFeps})), we construct an oracle which computes $\mathcal{D}_{n,m,h}[F(s,\cdot)](x)$ in (\ref{eq:DF}) by iteratively calling $O_F$ for $x+(-n)h,x+(-n+1)h,...,x+nh$.
Then, using QAE with this oracle, we can estimate $\mathcal{D}_{n,m,h}[V](x)$ in (\ref{eq:Vm}) as an approximation of $V^{(m)}(x)$.

However, there is the following subtlety in this way.
As mentioned above, it is common that $F$ is nonsmooth, and therefore $\mathcal{D}_{n,m,h}[F(s,\cdot)](x)$ is not bounded for small $h$.
For example, if $F(s,x)$ has a discontinuity with respect to $x$ at some $(s^\prime,x^\prime)\in\mathcal{S}\times \mathbb{R}$, $\mathcal{D}_{1,1,h}[F(s^\prime,\cdot)](x)=\frac{1}{2h}\left(F(s^\prime,x^\prime+h)-F(s^\prime,x^\prime-h)\right)$ diverges when $h\rightarrow0$.
On the other hand, as explained in Section \ref{sec:QAE}, in QMCI, we need to take some upper bound on the absolute values of the integrand and normalize it with the bound, in order to encode an integrand value into an amplitude of an ancilla qubit. 
This provides the difference between the cases where $F$ is smooth and nonsmooth.

\subsubsection{The case of a smooth integrand}

First, we consider the smooth integrand case.
In this case, $\mathcal{D}_{n,m,h}[F(s,\cdot)](x)$ is bounded, and therefore we can normalize it, and then estimate its expected value by QAE.
We give the formal statement on the quantum algorithm in this case as follows, presenting the concrete calculation procedure in the proof.

\begin{theorem}
	Let $\mathcal{S}$ be some finite set such that each $s\in\mathcal{S}$ is associated with $p_s\in[0,1]$ satisfying $\sum_{s\in\mathcal{S}}p_s =1$.
	Let $F$ be a real-valued function on $\mathcal{S}\times\mathbb{R}$ satisfying the following conditions:
	\begin{enumerate}
		\renewcommand{\labelenumi}{(\roman{enumi})}
		\item there exists $B\in\mathbb{R}$ such that $|F(s,x)|\le B$ holds for every $(s,x)\in\mathcal{S}\times \mathbb{R}$,
		\item there exist $A,c\in\mathbb{R}_+$ and $\sigma\in\mathbb{R}$ such that, for every $s\in\mathcal{S}$, $F(s,x)$ as a function of $x$ is in $\mathcal{G}_{A,c,\sigma}$.
	\end{enumerate}
	Suppose that we have an access to an oracle $O_S$, which performs the operation (\ref{eq:OraS}).
	Suppose that, for any given $\epsilon\in\mathbb{R}_+$, we have an access to an oracle $O_{F,\epsilon}$, which performs the operation (\ref{eq:OFeps}) for every $(s,x)\in\mathcal{S}\times\mathbb{R}$ with some function $F_{\epsilon}:\mathcal{S}\times \mathbb{R}\rightarrow\mathbb{R}$ satisfying (\ref{eq:FFepsdiff}) and uses qubits at most (\ref{eq:qubitOFeps}) with $a\in\mathbb{R}_+$.
	Suppose that we are given $x\in\mathbb{R}$, $m\in\mathbb{N}$ and $\epsilon\in\mathbb{R}_+$.
	Then, for any $n\in\mathbb{N}_{\ge \frac{m}{2}}$ and $h\in\mathbb{R}_+$ satisfying
	\begin{equation}
		Ac^{2n+1}((2n+1)!)^\sigma m \left(\frac{em}{2}\right)^{2n} h^{2n-m+1} \le \epsilon, \label{eq:hthSumInOra}
	\end{equation}
	there is a quantum algorithm $\mathcal{A}_1(m,\epsilon;n,h)$, which outputs $3\epsilon$-approximation of $V^{(m)}(x)$ with probability at least 0.99, making
	\begin{equation}
		O\left(\frac{Ac^m (m!)^\sigma}{\epsilon}\right) \label{eq:numOraSumInOracleS_SmF}
	\end{equation}
	calls to $O_S$ and
	\begin{equation}
		O\left(\frac{Ac^m (m!)^\sigma n}{\epsilon}\right)\label{eq:numOraSumInOracleF_SmF}
	\end{equation}
	calls to $O_{F,\tilde{\epsilon}}$ and using	qubits at most
	\begin{equation}
		O\left(\log^a\left(\frac{m(2(1+\log n))^m}{h^m\epsilon}\right)\right) \label{eq:qubitOraSumInOracle}
	\end{equation}
	for $O_{F,\tilde{\epsilon}}$.
	Here,
	\begin{equation}
		\tilde{\epsilon} := \frac{h^m\epsilon}{D^{(m)}_n}. \label{eq:epstil}
	\end{equation}
	In particular, $\mathcal{A}_1\left(m,\epsilon;\left\lceil \frac{m}{2} \right\rceil,h_{\rm min}\right)$, where
	\begin{equation}
		h_{\rm min} := 
		\begin{dcases}
			\left(\frac{\epsilon}{Ac((m+2)!)^\sigma m}\left(\frac{2}{ecm}\right)^{m+1}\right)^{1/2} & ; \ {\rm if} \ m \ {\rm is \ odd} \\
			\frac{\epsilon}{Ac((m+1)!)^\sigma m}\left(\frac{2}{ecm}\right)^{m} & ; \ {\rm if} \ m \ {\rm is \ even}
		\end{dcases}
		, \label{eq:h1}
	\end{equation}
	makes
	\begin{equation}
		O\left(\frac{Ac^m (m!)^\sigma m}{\epsilon}\right)\label{eq:numOraSumInOracleF_SmF_case1}
	\end{equation}
	calls to $O_{F,\tilde{\epsilon}}$ and uses
	\begin{fleqn}[-25pt]
		\begin{equation}
			\begin{dcases}
			O\left(\log^a\left(\frac{e^{\frac{m^2+m}{2}}m^{\frac{1}{2}m^2+m+1}A^{\frac{m}{2}}c^{\frac{1}{2}m^2+m}((m+2)!)^{\frac{m\sigma}{2}}\left[\left(1+\log \left(\frac{m+1}{2}\right)\right)\right]^m}{2^{\frac{m^2-m}{2}}\epsilon^{\frac{m}{2}+1}}\right)\right) \\
			 \qquad\qquad\qquad\qquad\qquad\qquad\qquad\qquad\qquad\qquad\qquad ; \ {\rm if} \ m \ {\rm is \ odd} \\
			O\left(\log^a\left(\frac{e^{m^2}m^{m^2+m+1}A^mc^{m^2+m}((m+1)!)^{m\sigma}\left[\left(1+\log \left(\frac{m}{2} \right)\right)\right]^m}{2^{m^2}\epsilon^{m+1}}\right)\right) \\
			 \qquad\qquad\qquad\qquad\qquad\qquad\qquad\qquad\qquad\qquad\qquad; \ {\rm if} \ m \ {\rm is \ even}
			\end{dcases}
			 \label{eq:qubitOraSumInOracle_case1}
		\end{equation}
	\end{fleqn}
	qubits for $O_{F,\tilde{\epsilon}}$, and $\mathcal{A}_1(m,\epsilon;n_{\rm th},h_{\rm th})$, where $n_{\rm th}$ and $h_{\rm th}$ are given as (\ref{eq:nth}) and (\ref{eq:hcond}) respectively, makes
	\begin{equation}
		O\left(\frac{Ac^m (m!)^\sigma m}{\epsilon}\log_2\left(\frac{2^{m\sigma^+}}{\epsilon^\prime}\right)\right) \label{eq:numOraSumInOracleF_SmF_case2}
	\end{equation}
	calls to $O_{F,\tilde{\epsilon}}$ and uses
	\begin{equation}
		O\left(\log^a\left(\frac{m(2ecm)^mB}{\epsilon}\right)\right) \label{eq:qubitOraSumInOracle_case2}
	\end{equation}
	qubits for $O_{F,\tilde{\epsilon}}$, where $\epsilon^\prime$ is given as (\ref{eq:epsprime}).
	\label{th:InOracleSmF} 
\end{theorem}

\begin{proof}
	We first present the algorithm, and then consider the accuracy and the complexity.\\
	
	\noindent \textbf{Algorithm}
	
	Consider a system consisting of five quantum registers $R_1,...,R_5$ and some ancillary registers as necessary.
	$R_1$ to $R_4$ have sufficient numbers of qubits, whereas $R_5$ has a single qubit.
		We can perform the following operation to the system initialized to $\ket{0}\ket{0}\ket{0}\ket{0}\ket{0}$:\\
	
	\begin{algorithm}[H]
		\caption{}
		\label{proc1}
		\begin{algorithmic}[1] 
			\STATE Using $O_S$, generate $\sum_{s\in\mathcal{S}} \sqrt{p_s} \ket{s}$ on $R_1$. 
			\FORALL {$i\in\mathcal{J}^{\ne0}_{n,m}:=\{j\in[-n:n] \ | \ d^{(m)}_{n,j}\ne0 \}$}
			\STATE Set $R_2$ to $\ket{x+ih}$.
			\STATE By $O_{F,\tilde{\epsilon}}$, compute $F_{\tilde{\epsilon}}(s,x+ih)$ onto $R_3$, using the values on $R_1$ and $R_2$ as inputs.
			\STATE By a multiplier circuit (e.g. \cite{AlvarezSanchez,Jayashree,MunozCoreas}), add the product of $d^{(m)}_{n,i}$ and the value on $R_3$ to $R_4$.
			\STATE By $O_{F,\tilde{\epsilon}}^{-1}$, uncompute $R_3$ to $\ket{0}$.
			\ENDFOR
			\STATE By arithmetic circuits and a controlled rotation gate, transform the state on $R_5$ to 
			\begin{equation}
				\sqrt{\frac{1}{2}+\frac{X}{2h^m(Ac^m (m!)^\sigma+2\epsilon)}}\Ket{1}+\sqrt{\frac{1}{2}-\frac{X}{2h^m(Ac^m (m!)^\sigma+2\epsilon)}}\Ket{0}.
			\end{equation}
			using the value on $R_4$ as $X$. 
		\end{algorithmic}
	\end{algorithm}
	We denote the oracle that corresponds to this operation as $Q_1$.
	In this operation, the quantum state is transformed as follows:
	\clearpage
	\begin{widetext}
		
		\begin{eqnarray}
			&&\ket{0}\ket{0}\ket{0}\ket{0}\ket{0} \nonumber\\
			&\xrightarrow{1}& \sum_{s\in\mathcal{S}} \sqrt{p_s} \ket{s}\ket{0}\ket{0}\ket{0}\ket{0} \nonumber \\
			&\xrightarrow{3 \ {\rm for} \ i=-n}& \sum_{s\in\mathcal{S}} \sqrt{p_s} \ket{s}\ket{x+(-n)h}\ket{0}\ket{0}\ket{0} \nonumber \\
			&\xrightarrow{4 \ {\rm for} \ i=-n}& \sum_{s\in\mathcal{S}} \sqrt{p_s} \ket{s}\ket{x+(-n)h}\ket{F_{\tilde{\epsilon}}(s,x+(-n)h)}\ket{0}\ket{0} \nonumber \\
			&\xrightarrow{5 \ {\rm for} \ i=-n}& \sum_{s\in\mathcal{S}} \sqrt{p_s} \ket{s}\ket{x+(-n)h}\ket{F_{\tilde{\epsilon}}(s,x+(-n)h)}\ket{d^{(m)}_{n,-n}F_{\tilde{\epsilon}}(s,x+(-n)h)}\ket{0} \nonumber \\
			&\xrightarrow{6 \ {\rm for} \ i=-n}& \sum_{s\in\mathcal{S}} \sqrt{p_s} \ket{s}\ket{x+(-n)h}\ket{0}\ket{d^{(m)}_{n,-n}F_{\tilde{\epsilon}}(s,x+(-n)h)}\ket{0} \nonumber \\
			&\xrightarrow{3 \ {\rm for} \ i=-n+1}& \sum_{s\in\mathcal{S}} \sqrt{p_s} \ket{s}\ket{x+(-n+1)h}\ket{0}\ket{d^{(m)}_{n,-n}F_{\tilde{\epsilon}}(s,x+(-n)h)}\ket{0} \nonumber \\
			&\xrightarrow{4 \ {\rm for} \ i=-n+1}& \sum_{s\in\mathcal{S}} \sqrt{p_s} \ket{s}\ket{x+(-n+1)h}\ket{F_{\tilde{\epsilon}}(s,x+(-n+1)h)}\ket{d^{(m)}_{n,-n}F_{\tilde{\epsilon}}(s,x+(-n)h)}\ket{0} \nonumber \\
			&\xrightarrow{5 \ {\rm for} \ i=-n+1}& \sum_{s\in\mathcal{S}} \sqrt{p_s} \ket{s}\ket{x+(-n+1)h}\ket{F_{\tilde{\epsilon}}(s,x+(-n+1)h)}\Ket{\sum_{j=-n}^{-n+1}d^{(m)}_{n,j}F_{\tilde{\epsilon}}(s,x+jh)}\ket{0} \nonumber \\
			&\xrightarrow{6 \ {\rm for} \ i=-n+1}& \sum_{s\in\mathcal{S}} \sqrt{p_s} \ket{s}\ket{x+(-n+1)h}\ket{0}\Ket{\sum_{j=-n}^{-n+1}d^{(m)}_{n,j}F_{\tilde{\epsilon}}(s,x+jh)}\ket{0} \nonumber \\
			&\rightarrow& \cdots \nonumber \\
			&\xrightarrow{6 \ {\rm for} \ i=n}& \sum_{s\in\mathcal{S}} \sqrt{p_s} \ket{s}\ket{x+nh}\ket{0}\Ket{\sum_{j=-n}^{n}d^{(m)}_{n,j}F_{\tilde{\epsilon}}(s,x+jh)}\ket{0} \nonumber \\
			&\xrightarrow{8}& \sum_{s\in\mathcal{S}} \sqrt{p_s} \ket{s}\ket{x+nh}\ket{0}\Ket{\sum_{j=-n}^{n}d^{(m)}_{n,j}F_{\tilde{\epsilon}}(s,x+jh)}\nonumber \\
			&&\qquad \ \ \otimes\left(\sqrt{\frac{1}{2}+\frac{1}{2h^m(Ac^m (m!)^\sigma+2\epsilon)}\sum_{j=-n}^{n}d^{(m)}_{n,j}F_{\tilde{\epsilon}}(s,x+jh)}\Ket{1}+\sqrt{\frac{1}{2}-\frac{1}{2h^m(Ac^m (m!)^\sigma+2\epsilon)}\sum_{j=-n}^{n}d^{(m)}_{n,j}F_{\tilde{\epsilon}}(s,x+jh)}\Ket{0}\right)\nonumber \\
			&=:&\ket{\Psi_1}, \label{eq:transfSumInOracle_SmF}
		\end{eqnarray}
	\end{widetext}
	Note that the insides of the square roots in the last line in (\ref{eq:transfSumInOracle_SmF}) are in $[0,1]$, since
	\begin{fleqn}[-10pt]
	\begin{eqnarray}
		&&\left|\frac{1}{h^m(Ac^m (m!)^\sigma+2\epsilon)}\sum_{j=-n}^{n}d^{(m)}_{n,j}F_{\tilde{\epsilon}}(s,x+jh)\right| \nonumber \\
		&\le& \left|\frac{1}{h^m(Ac^m (m!)^\sigma+2\epsilon)}\sum_{j=-n}^{n}d^{(m)}_{n,j}F(s,x+jh)\right| + \nonumber \\
		&&\quad \left|\frac{1}{h^m(Ac^m (m!)^\sigma+2\epsilon)}\sum_{j=-n}^{n}d^{(m)}_{n,j}\left(F(s,x+jh)-F_{\tilde{\epsilon}}(s,x+jh)\right)\right|  \nonumber \\
		&\le& \frac{\left|\mathcal{D}_{n,m,h}[F(s,\cdot)](x)\right|}{Ac^m (m!)^\sigma+2\epsilon}+ \nonumber \\
		&& \quad \frac{1}{h^m(Ac^m (m!)^\sigma+2\epsilon)} \sum_{j=-n}^{n}\left|d^{(m)}_{n,j}\right|\cdot\left|F(s,x+jh)-F_{\tilde{\epsilon}}(s,x+jh)\right|   \nonumber \\
		&\le& \frac{Ac^m (m!)^\sigma+\epsilon}{Ac^m (m!)^\sigma+2\epsilon} + \frac{1}{h^m(Ac^m (m!)^\sigma+2\epsilon)} D^{(m)}_n \tilde{\epsilon} \nonumber \\
		&=& 1.
	\end{eqnarray}
	\end{fleqn}
	Here, at the third inequality, we used
	\begin{eqnarray}
		&&\left|\mathcal{D}_{n,m,h}[F(s,\cdot)](x)\right| \nonumber \\
		&\le& \left|\frac{\partial^m F(s,x)}{\partial x^m}\right|+\left|\mathcal{D}_{n,m,h}[F(s,\cdot)](x)-\frac{\partial^m F(s,x)}{\partial x^m}\right|\nonumber \\
		&\le& Ac^m (m!)^\sigma+\epsilon, \nonumber
	\end{eqnarray}
	which follows from $F(s,\cdot)\in\mathcal{G}_{A,c,\sigma}$, Lemma \ref{lem:R} and (\ref{eq:hthSumInOra}).
	The probability that we obtain $1$ on the last qubit in measuring $\ket{\Psi_1}$ is
	\begin{equation}
		P:=\frac{1}{2}+\frac{1}{2h^m(Ac^m (m!)^\sigma+2\epsilon)}\sum_{j=-n}^{n}d^{(m)}_{n,j}\sum_{s\in\mathcal{S}} p_s F_{\tilde{\epsilon}}(s,x+jh).
	\end{equation}
	Defining
	\begin{equation}
		Y:=(Ac^m (m!)^\sigma+2\epsilon)(2P-1), \label{eq:outAlg1til}
	\end{equation}
	we see that
	\begin{eqnarray}
		&&\left|\mathcal{D}_{n,m,h}[V](x)-Y\right| \nonumber \\
		&\le& \sum_{j=-n}^{n}\sum_{s\in\mathcal{S}} \frac{1}{h^m}\left|d^{(m)}_{n,j}\right| p_s \left|F(s,x+jh)-F_{\tilde{\epsilon}}(s,x+jh)\right| \nonumber \\
		&\le& \frac{1}{h^m}D^{(m)}_n\tilde{\epsilon}  \nonumber \\
		&=& \epsilon. \label{eq:DVY2}
	\end{eqnarray}
	Therefore, we obtain an estimate of $\mathcal{D}_{n,m,h}[V](x)$ as follows: obtain an estimate $\tilde{P}$ of $P$ by QAE, in which $Q_1$ is iteratively called, and then output
	\begin{equation}
		\tilde{Y}:=(Ac^m (m!)^\sigma+2\epsilon)(2\tilde{P}-1). \label{eq:tilY2}
	\end{equation}
	\\
	
	\noindent \textbf{Accuracy and complexity}
	
	\begin{eqnarray}
		&&|V^{(m)}(x)-\mathcal{D}_{n,m,h}[V](x)| \nonumber \\
		&\le& |R_V(x,n,m,h)| h^{2n-m+1} \nonumber \\
		&\le& Ac^{2n+1}((2n+1)!)^\sigma m \left(\frac{em}{2}\right)^{2n} h^{2n-m+1} \nonumber \\
		&\le& \epsilon \label{eq:VmDV}
	\end{eqnarray}
	holds.
	Here, the first inequality holds because of Theorem \ref{th:numDiff}.
	At the second inequality, we use Lemma \ref{lem:R} with $M= Ac^{2n+1}((2n+1)!)^\sigma$, since $V\in\mathcal{G}_{A,c,\sigma}$ as easily seen from $F(s,\cdot)\in\mathcal{G}_{A,c,\sigma}$ for every $s\in\mathcal{S}$.
	The last inequality is (\ref{eq:hthSumInOra}).
	Using (\ref{eq:VmDV}) and (\ref{eq:DVY2}), we see that, if we have $\tilde{Y}$ such that
	\begin{equation}
		|\tilde{Y}-Y|\le \epsilon, \label{eq:YtilY}
	\end{equation}
	the following holds:
	\begin{eqnarray}
		&&|V^{(m)}(x)-\tilde{Y}| \nonumber \\
		&\le& |V^{(m)}(x)-\mathcal{D}_{n,m,h}[V](x)| + |\mathcal{D}_{n,m,h}[V](x)-Y| + |Y-\tilde{Y}| \nonumber \\
		&\le& 3\epsilon, \label{eq:VmY}
	\end{eqnarray}
	which means that $\tilde{Y}$ is an $3\epsilon$-approximation of $V^{(m)}(x)$.
	
	Then, let us estimate the query complexity to obtain $\tilde{P}$ such that (\ref{eq:YtilY}) holds by QAE.
	Because of the definitions (\ref{eq:outAlg1til}) and (\ref{eq:tilY2}), it is sufficient to obtain $\tilde{P}$ such that
	\begin{equation}
		|\tilde{P}-P|\le \frac{\epsilon}{2(Ac^m (m!)^\sigma+2\epsilon)}
	\end{equation}
	by QAE.
	For this, QAE with $N_{Q_1}$ calls to $Q_1$, where $N_{Q_1}$ is at most (\ref{eq:numOraSumInOracleS_SmF}), is sufficient.
	Since $Q_1$ uses $O_S$ once and $O_{F,\tilde{\epsilon}}$ at most $2n+1$ times, we evaluate the numbers of queries to them as (\ref{eq:numOraSumInOracleS_SmF}) and (\ref{eq:numOraSumInOracleF_SmF}).
	We also have (\ref{eq:qubitOraSumInOracle}), combining the assumption that $O_{F,\tilde{\epsilon}}$ uses $O\left(\log^a \left(\frac{1}{\tilde{\epsilon}}\right)\right)$ qubits with (\ref{eq:epstil}) and Lemma \ref{lem:cj}.\\
	
	For $n=\left\lceil\frac{m}{2}\right\rceil$ and $h=h_{\rm min}$, which we can check satisfy (\ref{eq:hthSumInOra}) by simple algebra, we just plug these values into (\ref{eq:numOraSumInOracleF_SmF}) and (\ref{eq:qubitOraSumInOracle}), and then obtain (\ref{eq:numOraSumInOracleF_SmF_case1}) and (\ref{eq:qubitOraSumInOracle_case1}), respectively (note that the number of queries to $O_S$ does not depend on $n$ and $h$).
	
	Also for $n=n_{\rm th}$ and $h=h_{\rm th}$, for which (\ref{eq:hthSumInOra}) holds as shown in the proof of Lemma \ref{lem:cenAppOrd} (see (\ref{eq:temp4})), just plugging these values into (\ref{eq:numOraSumInOracleF_SmF}) and (\ref{eq:qubitOraSumInOracle}) with some algebra leads to (\ref{eq:numOraSumInOracleF_SmF_case2}) and (\ref{eq:qubitOraSumInOracle_case2}), respectively.
	
\end{proof}

(\ref{eq:numOraSumInOracleS_SmF}) indicates that, asymptotically, the upper bound on the number of queries to $O_S$ scales $\epsilon$ as $O\left(\frac{1}{\epsilon}\right)$ and independent from $n$ and $h$.
Besides, from (\ref{eq:numOraSumInOracleF_SmF}), we see that the query number bound for $O_{F,\tilde{\epsilon}}$ depends on not $h$ but $n$ linearly, and scales as $O\left(\frac{1}{\epsilon}\right)$ if $n$ does not depend on $\epsilon$.
Hence, setting $n$ to the minimum value $\left\lceil\frac{m}{2}\right\rceil$ and $h$ to (\ref{eq:h1}) is best in the aspect of this bound.
However, the qubit number for $O_{F,\tilde{\epsilon}}$ becomes large in this setting, depending on $\epsilon$ as $O\left(\log^a\left(\frac{1}{\epsilon^{\frac{m}{2}+1}}\right)\right)$ or $O\left(\log^a\left(\frac{1}{\epsilon^{m+1}}\right)\right)$.
On the other hand, setting $n=n_{\rm th}$ and $h=h_{\rm th}$ leads to less qubit number scaling as $O\left(\log^a\left(\frac{1}{\epsilon}\right)\right)$, adding a $O\left(\log\left(\frac{1}{\epsilon}\right)\right)$ factor to the query number bound for $O_{F,\tilde{\epsilon}}$.

\subsubsection{The case of a nonsmooth integrand}

Next, we consider the nonsmooth integrand case.
Now, $\mathcal{D}_{n,m,h}[F(s,\cdot)](x)$ can be unbounded when $h\rightarrow 0$, and therefore, in the algorithm we propose, we estimate the expectation value of not $\mathcal{D}_{n,m,h}[F(s,\cdot)](x)$ but $\sum_{j=-n}^{n} d^{(m)}_{n,j} F(s,x+jh)$, omitting the factor $1/h^{m}$, and then divide the estimate by $h^{m}$ to obtain $\mathcal{D}_{n,m,h}[V](x)$.
We present the formal statement on this method as follows.

\begin{theorem}

Let $\mathcal{S},x,m$ and $\epsilon$ be as described in Theorem \ref{th:InOracleSmF}.
Let $F:\mathcal{S}\times\mathbb{R}\rightarrow\mathbb{R}$ be a function satisfying the condition (i) in Theorem \ref{th:InOracleSmF} and the following
\begin{enumerate}
	\renewcommand{\labelenumi}{(\roman{enumi})'}
	\setcounter{enumi}{1}
	\item there exist $A,c\in\mathbb{R}_+$ and $\sigma\in\mathbb{R}$ such that $V$ defined as (\ref{eq:V}) is in $\mathcal{G}_{A,c,\sigma}$.
\end{enumerate}
Suppose that we are given accesses to the oracles $O_S$ and $O_{F,\epsilon}$ for any $\epsilon\in\mathbb{R}_+$, which are described in Theorem \ref{th:InOracleSmF}.
Then, for any $n\in\mathbb{N}_{\ge \frac{m}{2}}$ and $h\in\mathbb{R}_+$ satisfying (\ref{eq:hthSumInOra}), there is a quantum algorithm $\mathcal{A}_2(m,\epsilon;n,h)$, which outputs $3\epsilon$-approximation of $V^{(m)}(x)$ with probability at least 0.99, making
\begin{equation}
	O\left(\frac{m\left[2\left(1+\log n\right)\right]^m B}{h^m\epsilon}\right) \label{eq:numOraSumInOracleS}
\end{equation}
calls to $O_S$ and
\begin{equation}
	O\left(\frac{mn\left[2\left(1+\log n\right)\right]^m B}{h^m\epsilon}\right)\label{eq:numOraSumInOracleF}
\end{equation}
calls to $O_{F,\tilde{\epsilon}}$ and using qubits at most (\ref{eq:qubitOraSumInOracle}) for $O_{F,\tilde{\epsilon}}$.
Here, $\tilde{\epsilon}$ is given as (\ref{eq:epstil}).
In particular, $\mathcal{A}_2\left(m,\epsilon;\left\lceil \frac{m}{2} \right\rceil,h_{\rm min}\right)$, where $h_{\rm min}$ is given as (\ref{eq:h1}), makes
\begin{fleqn}[-35pt]
\begin{equation}
	\begin{dcases}
	O\left(\frac{e^{\frac{1}{2}m(m+1)}m^{\frac{1}{2}m^2+m+1}A^{\frac{m}{2}}c^{\frac{1}{2}m(m+2)}((m+2)!)^{m\sigma/2}\left[2\left(1+\log \left(\frac{m+1}{2}\right)\right)\right]^m B}{2^{\frac{1}{2}m(m+1)}\epsilon^{\frac{m}{2}+1}}\right) \\
	\qquad\qquad\qquad\qquad\qquad\qquad\qquad\qquad\qquad\qquad\qquad;  \ {\rm if} \ m \ {\rm is \ odd} \\
	O\left(\frac{e^{m^2}m^{m^2+m+1}A^mc^{m^2+m}((m+1)!)^{m\sigma}\left[2\left(1+\log \left(\frac{m}{2}\right)\right)\right]^m B}{2^{m^2}\epsilon^{m+1}}\right) \\
	\qquad\qquad\qquad\qquad\qquad\qquad\qquad\qquad\qquad\qquad\qquad ;  \ {\rm if} \ m \ {\rm is \ even}
	\end{dcases} \label{eq:numOraSumInOracleS_case1}
\end{equation}
\end{fleqn}
calls to $O_S$ and
\begin{fleqn}[-35pt]
\begin{equation}
	\begin{dcases}
		O\left(\frac{e^{\frac{1}{2}m(m+1)}m^{\frac{1}{2}m^2+m+2}A^{\frac{m}{2}}c^{\frac{1}{2}m(m+2)}((m+2)!)^{m\sigma/2}\left[2\left(1+\log \left(\frac{m+1}{2}\right)\right)\right]^m B}{2^{\frac{1}{2}m(m+1)}\epsilon^{\frac{m}{2}+1}}\right) \\
		\qquad\qquad\qquad\qquad\qquad\qquad\qquad\qquad\qquad\qquad\qquad;  \ {\rm if} \ m \ {\rm is \ odd} \\
		O\left(\frac{e^{m^2}m^{m^2+m+2}A^mc^{m^2+m}((m+1)!)^{m\sigma}\left[2\left(1+\log \left(\frac{m}{2}\right)\right)\right]^m B}{2^{m^2}\epsilon^{m+1}}\right) \\
		\qquad\qquad\qquad\qquad\qquad\qquad\qquad\qquad\qquad\qquad\qquad ;  \ {\rm if} \ m \ {\rm is \ even}
	\end{dcases}
	\label{eq:numOraSumInOracleF_case1}
\end{equation}
\end{fleqn}
calls to $O_{F,\tilde{\epsilon}}$ and uses qubits at most (\ref{eq:qubitOraSumInOracle_case1}) for $O_{F,\tilde{\epsilon}}$, and $\mathcal{A}_2(m,\epsilon;n_{\rm th},h_{\rm th})$, where $n_{\rm th}$ and $h_{\rm th}$ are given as (\ref{eq:nth}) and (\ref{eq:hcond}) respectively, makes
\begin{equation}
	O\left(\frac{m(2ecm)^mB}{\epsilon}\log_2^{m\sigma^+}\left(\frac{2^{m\sigma^+}}{\epsilon^\prime}\right)\log^m\left(\log_2\left(\frac{2^{m\sigma^+}}{\epsilon^\prime}\right)\right)\right) \label{eq:numOraSumInOracleS_case2}
\end{equation}
calls to $O_S$ and
\begin{equation}
	O\left(\frac{m(2ecm)^mB}{\epsilon}\log_2^{m\sigma^+ + 1}\left(\frac{2^{m\sigma^+}}{\epsilon^\prime}\right)\log^m\left(\log_2\left(\frac{2^{m\sigma^+}}{\epsilon^\prime}\right)\right)\right) \label{eq:numOraSumInOracleF_case2}
\end{equation}
calls to $O_{F,\tilde{\epsilon}}$ and uses qubits at most (\ref{eq:qubitOraSumInOracle_case2}) for $O_{F,\tilde{\epsilon}}$, where $\epsilon^\prime$ is given as (\ref{eq:epsprime}).
\label{th:InOracle}
\end{theorem}

\begin{proof}
	We first present the algorithm, and then consider the accuracy and the complexity.\\
	
	\noindent \textbf{Algorithm}
	
	Consider a system consisting of same quantum registers $R_1,...,R_5$ as Theorem \ref{th:InOracleSmF} and some ancillary registers as necessary.
	We can perform the operation same as Procedure \ref{proc1} except the replacement of (\ref{eq:ancRot}) with
	\begin{equation}
		\sqrt{\frac{1}{2}+\frac{X}{2(D^{(m)}_nB+\tilde{\epsilon})}}\Ket{1}+\sqrt{\frac{1}{2}-\frac{X}{2(D^{(m)}_nB+\tilde{\epsilon})}}\Ket{0}, \label{eq:ancRot}
	\end{equation}
	We denote the oracle that corresponds to this operation as $Q_2$.
	By this operation, the quantum state is transformed as follows: 
	\begin{widetext}
		
		\begin{eqnarray}
			&&\ket{0}\ket{0}\ket{0}\ket{0}\ket{0} \nonumber\\
			&\rightarrow& \sum_{s\in\mathcal{S}} \sqrt{p_s} \ket{s}\ket{x+nh}\ket{0}\Ket{\sum_{j=-n}^{n}d^{(m)}_{n,j}F_{\tilde{\epsilon}}(s,x+jh)}\nonumber \\
			&&\qquad\qquad \otimes\left(\sqrt{\frac{1}{2}+\frac{1}{2D^{(m)}_n(B+\tilde{\epsilon})}\sum_{j=-n}^{n}d^{(m)}_{n,j}F_{\tilde{\epsilon}}(s,x+jh)}\Ket{1}+\sqrt{\frac{1}{2}-\frac{1}{2D^{(m)}_n(B+\tilde{\epsilon})}\sum_{j=-n}^{n}d^{(m)}_{n,j}F_{\tilde{\epsilon}}(s,x+jh)}\Ket{0}\right)=:\ket{\Psi_2}.\nonumber \\
			&& \label{eq:transfSumInOracle}
		\end{eqnarray}
	\end{widetext}
	Note that the insides of the square roots in the last line in (\ref{eq:transfSumInOracle}) are in $[0,1]$, since
	\begin{eqnarray}
		&&\left|\frac{1}{D^{(m)}_n(B+\tilde{\epsilon})}\sum_{j=-n}^{n}d^{(m)}_{n,j}F_{\tilde{\epsilon}}(s,x+jh)\right| \nonumber \\
		&\le& \left|\frac{1}{D^{(m)}_n(B+\tilde{\epsilon})}\sum_{j=-n}^{n}d^{(m)}_{n,j}F(s,x+jh)\right| \nonumber \\
		&&\qquad +\left|\frac{1}{D^{(m)}_n(B+\tilde{\epsilon})}\sum_{j=-n}^{n}d^{(m)}_{n,j}\left(F(s,x+jh)-F_{\tilde{\epsilon}}(s,x+jh)\right)\right|  \nonumber \\
		&\le& \frac{1}{D^{(m)}_n(B+\tilde{\epsilon})} \sum_{j=-n}^{n}\left|d^{(m)}_{n,j}\right|\cdot\left|F(s,x+jh)\right|  \nonumber \\
		&& \qquad + \frac{1}{D^{(m)}_n(B+\tilde{\epsilon})} \sum_{j=-n}^{n}\left|d^{(m)}_{n,j}\right|\cdot\left|F(s,x+jh)-F_{\tilde{\epsilon}}(s,x+jh)\right|   \nonumber \\
		&\le& \frac{1}{D^{(m)}_n(B+\tilde{\epsilon})} D^{(m)} B + \frac{1}{D^{(m)}_n(B+\tilde{\epsilon})} D^{(m)} \tilde{\epsilon} \nonumber \\
		&\le& 1.
	\end{eqnarray}
	The probability that we obtain $1$ on the last qubit in measuring $\ket{\Psi_2}$ is
	\begin{equation}
		P:=\frac{1}{2}+\frac{1}{2D^{(m)}_n(B+\tilde{\epsilon})}\sum_{j=-n}^{n}d^{(m)}_{n,j}\sum_{s\in\mathcal{S}} p_s F_{\tilde{\epsilon}}(s,x+jh). \label{eq:P}
	\end{equation}
	Defining
	\begin{equation}
		Y:=\frac{D^{(m)}_n (B+\tilde{\epsilon})}{h^m}(2P-1), \label{eq:outAlg1}
	\end{equation}
	we see that $\left|\mathcal{D}_{n,m,h}[V](x)-Y\right|\le\epsilon$ similarly to (\ref{eq:DVY2}).
	Therefore, we obtain an estimate of $\mathcal{D}_{n,m,h}[V](x)$ as follows: obtain an estimate $\tilde{P}$ of $P$ by QAE, in which $Q_2$ is iteratively called, and then output
	\begin{equation}
		\tilde{Y}:=\frac{D^{(m)}_n (B+\tilde{\epsilon})}{h^m}(2\tilde{P}-1). \label{eq:tilY}
	\end{equation}
	\\
	
	\noindent \textbf{Accuracy and complexity}
	
	As shown in the proof of Theorem \ref{th:InOracleSmF}, if we have $\tilde{Y}$ such that $|\tilde{Y}-Y|\le \epsilon$, $\tilde{Y}$ is an $3\epsilon$-approximation of $V^{(m)}(x)$.	
	Then, let us estimate the query complexity to obtain $\tilde{P}$ that makes this hold by QAE.
	Because of the definitions (\ref{eq:outAlg1}) and (\ref{eq:tilY}), it is sufficient to obtain $\tilde{P}$ such that
	\begin{equation}
	  	|\tilde{P}-P|\le \frac{h^m}{2D^{(m)}_n (B+\tilde{\epsilon})}\epsilon
	\end{equation}
	by QAE.
	For this, QAE with $N_{Q_2}$ calls to $Q_2$, where
	\begin{equation}
		N_{Q_2}=O\left(\frac{D^{(m)}_n (B+\tilde{\epsilon})}{h^m\epsilon}\right), \label{eq:compTemp}
	\end{equation}
	is sufficient.
	Using Lemma \ref{lem:cj} and (\ref{eq:hthSumInOra}) with simple algebra, we see that (\ref{eq:compTemp}) is evaluated as (\ref{eq:numOraSumInOracleS}).
	Since $Q_2$ uses $O_S$ once and $O_{F,\tilde{\epsilon}}$ at most $2n+1$ times, we have (\ref{eq:numOraSumInOracleS}) and (\ref{eq:numOraSumInOracleF}).
	We also prove the claim on the qubit number for $O_{F,\tilde{\epsilon}}$ similarly to Theorem \ref{th:InOracleSmF}.\\
	
	Remaining claims for the specific settings on $n$ and $h$ are proven by simply plugging their values into the expressions of the query numbers and the qubit number.
	
\end{proof}

Since we multiply the result of QAE by $\frac{D^{(m)}_n (B+\tilde{\epsilon})}{h^m}$, this factor is included also in the upper bounds (\ref{eq:numOraSumInOracleS}) and (\ref{eq:numOraSumInOracleF}) on the query numbers in this method.
In the case of $n=\left\lceil\frac{m}{2}\right\rceil$, the minimum value of $n$, this increases the exponent of $\frac{1}{\epsilon}$ in the query number bounds, since we set $h$ to a small value depending on $\epsilon$ as (\ref{eq:h1}).
Conversely, if we set $n$ efficiently large as $n=n_{\rm th}$, we can set $h$ to $h_{\rm th}$, which depends on $\epsilon$ only logarithmically through $n$, and therefore the query number bounds (\ref{eq:numOraSumInOracleS_case2}) and (\ref{eq:numOraSumInOracleF_case2}) scale with $\epsilon$ as $O\left(\frac{1}{\epsilon}\right)$, expect the logarithmic factor.
Besides, larger $n$ and $h$ lead to the smaller qubit number for $O_{F,\epsilon}$, similarly to the smooth integrand case.
Hence, in the nonsmooth integrand case, setting $n=n_{\rm th}$ and $h=h_{\rm th}$ is better than setting to smaller numbers, in terms of both query complexity and qubit number.

\subsection{\label{sec:algoDetail} The sum-in-QAE method}

We now consider another quantum method for numerical differentiation of $V$, which we name the {\it sum-in-QAE} method.
Note that (\ref{eq:Vm}) is a two-fold summation, which consists of the sum over the values $s$ of the stochastic variable and the sum over $j$ in the finite difference formula.
Then, we can take the latter sum at the same time as the former in one QAE, unlike the iterative calls to $O_{F,\epsilon}$ for $x+(-n)h,...,x+nh$ in the naive iteration method.
We present the detail of this method in the following theorem.

\begin{theorem}
	Let $\mathcal{S},x,m,\epsilon$ and $F$ be as described in Theorem \ref{th:InOracle}.
	Suppose that we are given accesses to the oracles $O_S$ and $O_{F,\epsilon}$ for any $\epsilon\in\mathbb{R}_+$, which are described in Theorem \ref{th:InOracleSmF} and \ref{th:InOracle}.
	Besides, suppose that, for every $n\in N_{\frac{m}{2}}$, we have accesses to the oracles $O_{\rm coef}^{m,n}$, which performs the operation
	\begin{equation}
		O_{\rm coef}^{m,n}\ket{0}=\ket{\Psi_{\rm coef}^{n,m}}:=\frac{1}{\sqrt{D^{(m)}_{n}}}\sum_{j\in[-n:n]} \sqrt{\left|d^{(m)}_{n,j}\right|} \ket{j}, \label{eq:CoefState}
	\end{equation}
	with $d^{(m)}_{n,j}$ defined as (\ref{eq:centGen}), and $O_{\rm sign}^{m,n}$, which performs the operation
	\begin{equation}
		O_{\rm sign}^{m,n}\ket{j}\ket{0}=\ket{j}\Ket{\theta^{m,n}_j}, \label{eq:signOra}
	\end{equation}
	for every $j\in[-n:n]$ with
	\begin{equation}
		\theta^{m,n}_j :=
		\begin{cases}
			1 & ; \ {\rm if} \ d^{(m)}_{n,j}\ge 0\\
			0 & ; \ {\rm otherwise}
		\end{cases}
	\end{equation}
	Then, for any $n\in\mathbb{N}_{\ge \frac{m}{2}}$ and $h\in\mathbb{R}_+$ satisfying (\ref{eq:hthSumInOra}), there is a quantum algorithm $\mathcal{A}_3(m,\epsilon;n,h)$, which outputs $3\epsilon$-approximation of $V^{(m)}(x)$ with probability at least 0.99, making	calls to $O_S$, $O_{F,\tilde{\epsilon}}$, $O_{\rm coef}^{m,n}$ and $O_{\rm sign}^{m,n}$ the number of times shown in (\ref{eq:numOraSumInOracleS}), and using qubits at most (\ref{eq:qubitOraSumInOracle}) for $O_{F,\tilde{\epsilon}}$, with $\tilde{\epsilon}$ given as (\ref{eq:epstil}).
	In particular, $\mathcal{A}_3(m,\epsilon;n_{\rm th},h_{\rm th})$, where $n_{\rm th}$ and $h_{\rm th}$ are given as (\ref{eq:nth}) and (\ref{eq:hcond}) respectively, calls $O_S$, $O_{F,\tilde{\epsilon}}$, $O_{\rm coef}^{m,n}$ and $O_{\rm sign}^{m,n}$ the number of times shown in (\ref{eq:numOraSumInOracleS_case2}), and uses qubits at most (\ref{eq:qubitOraSumInOracle_case2}) for $O_{F,\tilde{\epsilon}}$, where $\epsilon^\prime$ is given as (\ref{eq:epsprime}).
	\label{th:InQAE} 
\end{theorem}

\begin{proof}
	
Consider a system consisting of six quantum registers $R_1,...,R_6$ and some ancillary registers as necessary.
$R_3$ and $R_6$ have a single qubit, and the others have sufficient numbers of qubits.
We can perform the following operation to the system initialized to $\ket{0}\ket{0}\ket{0}\ket{0}\ket{0}\ket{0}$:\\

\begin{algorithm}[H]
	\caption{}
	\label{proc2}
	\begin{algorithmic}[1]
		\STATE Using $O_{\rm coef}^{m,n}$, generate the state (\ref{eq:CoefState}) on $R_1$.
		\STATE Set $x+jh$ on $R_2$, using the value on $R_1$ as $j$.
		\STATE By $O_{\rm sign}^{m,n}$, set $\theta^{m,n}_j$ on $R_3$, using the value on $R_1$ as $j$.
		\STATE Using $O_S$, generate $\sum_{s\in\mathcal{S}} \sqrt{p_s} \ket{s}$ on $R_4$. 
		\STATE By $O_{F,\tilde{\epsilon}}$, compute $F_{\tilde{\epsilon}}(s,x+jh)$ onto $R_5$, using the values on $R_2$ and $R_4$ as inputs.
		\STATE By a circuit similar to that in step 8 in Procedure \ref{proc1} and a NOT gate on $R_6$ activated only if the value on $R_3$ is 0, transform the state on $R_6$ to 
		\begin{equation}
			\sqrt{\frac{1}{2}+\frac{X}{2(B+\tilde{\epsilon})}}\ket{\theta}+\sqrt{\frac{1}{2}-\frac{X}{2(B+\tilde{\epsilon})}}\ket{1-\theta}, \label{eq:ancRot2}
		\end{equation}
		where $\theta$ and $X$ are the values on $R_3$ and $R_5$, respectively. 
	\end{algorithmic}
\end{algorithm}
We denote the oracle that corresponds to this operation by $Q_3$.
In this operation, the quantum state is transformed as follows: 
	\begin{widetext}
		
		\begin{eqnarray}
			&&\ket{0}\ket{0}\ket{0}\ket{0}\ket{0}\ket{0} \nonumber\\
			&\xrightarrow{1}& \frac{1}{\sqrt{D^{(m)}_{n}}}\sum_{j\in[-n:n]} \sqrt{\left|d^{(m)}_{n,j}\right|} \ket{j}\ket{0}\ket{0}\ket{0}\ket{0}\ket{0} \nonumber \\
			&\xrightarrow{2 \ {\rm and} \ 3}& \frac{1}{\sqrt{D^{(m)}_{n}}}\sum_{j\in[-n:n]} \sqrt{\left|d^{(m)}_{n,j}\right|} \ket{j}\ket{x+jh}\ket{\theta^{m,n}_j}\ket{0}\ket{0}\ket{0} \nonumber \\
			&\xrightarrow{4}& \frac{1}{\sqrt{D^{(m)}_{n}}}\sum_{j\in[-n:n]}\sum_{s\in\mathcal{S}} \sqrt{\left|d^{(m)}_{n,j}\right|p_s} \ket{j}\ket{x+jh}\ket{\theta^{m,n}_j}\ket{s}\ket{0}\ket{0} \nonumber \\
			&\xrightarrow{5}& \frac{1}{\sqrt{D^{(m)}_{n}}}\sum_{j\in[-n:n]}\sum_{s\in\mathcal{S}} \sqrt{\left|d^{(m)}_{n,j}\right|p_s} \ket{j}\ket{x+jh}\ket{\theta^{m,n}_j}\ket{s}\ket{F_{\tilde{\epsilon}}(s,x+jh)}\ket{0} \nonumber \\
			&\xrightarrow{6}& \frac{1}{\sqrt{D^{(m)}_{n}}}\sum_{j\in[-n:n]}\sum_{s\in\mathcal{S}} \sqrt{\left|d^{(m)}_{n,j}\right|p_s} \ket{j}\ket{x+jh}\ket{\theta^{m,n}_j}\ket{s}\ket{F_{\tilde{\epsilon}}(s,x+jh)}\nonumber \\
			&&\qquad\qquad\qquad\qquad \otimes\left(\sqrt{\frac{1}{2}+\frac{F_{\tilde{\epsilon}}(s,x+jh)}{2(B+\tilde{\epsilon})}}\Ket{\theta^{m,n}_j}+\sqrt{\frac{1}{2}-\frac{F_{\tilde{\epsilon}}(s,x+jh)}{2(B+\tilde{\epsilon})}}\Ket{1-\theta^{m,n}_j}\right)=:\ket{\Phi}.\nonumber \\
			&& \label{eq:mainTransf}
		\end{eqnarray}
	\end{widetext}
	Note that the insides of the square roots in the last line in (\ref{eq:mainTransf}) are in $[0,1]$, since
	\begin{equation}
		|F_{\tilde{\epsilon}}(s,x+jh)|\le|F(s,x+jh)|+|F_{\tilde{\epsilon}}(s,x+jh)-F(s,x+jh)|\le B+\tilde{\epsilon}. \nonumber
	\end{equation}
	The probability that we obtain $1$ on $R_6$ in measuring $\ket{\Phi}$ is equal to $P$ in (\ref{eq:P}), and $Y$ defined as (\ref{eq:outAlg1}) satisfies $\left|\mathcal{D}_{n,m,h}[V](x)-Y\right|\le\epsilon$ as seen in the proof of Theorem \ref{th:InOracle}.
	Therefore, we get an estimate of $\mathcal{D}_{n,m,h}[V](x)$ as follows: obtain an estimate $\tilde{P}$ of $P$ by QAE, in which $Q_3$ is iteratively called, and then output $\tilde{Y}$ in (\ref{eq:tilY}).

	\quad \\
	
	\noindent \textbf{Accuracy and complexity}
	
	The sum-in-QAE method is same as the naive iteration method for nonsmooth $F$ in that the quantity estimated by QAE is $P$ in (\ref{eq:P}) and that the final output is $\tilde{Y}$ in (\ref{eq:tilY}) calculated with the estimate $\tilde{P}$ for $P$, although the iteratively called oracles $Q_2$ and $Q_3$ are different.
	Therefore, the number of calls to $Q_3$ in the sum-in-QAE method is equal to the number of calls to $Q_2$ in the naive iteration method for nonsmooth $F$, and is evaluated as (\ref{eq:compTemp}), or (\ref{eq:numOraSumInOracleS}).
	Since $O_S$, $O_{F,\tilde{\epsilon}}$, $O_{\rm coef}^{m,n}$ and $O_{\rm sign}^{m,n}$ are called once in $Q_3$, the numbers of calls to these oracles are also evaluated as (\ref{eq:numOraSumInOracleS}).
	Since the same $\tilde{\epsilon}$ is used in the two methods, the qubit number for $O_{F,\tilde{\epsilon}}$ is also same, and given as (\ref{eq:qubitOraSumInOracle}).
	
	The remaining claim on the specific case\footnote{Although we considered the case that $n$ is set to the minimum value $\left\lceil\frac{m}{2}\right\rceil$ in Theorem \ref{th:InOracle}, we omit this case here, since this is less efficient than larger $n$ in terms of both query complexity and qubit number, as we saw in the proof of Theorem \ref{th:InOracle}.} that $n=n_{\rm th}$ and $h=h_{\rm th}$ is also proven similarly to the discussion in the proof of Theorem \ref{th:InOracle}.

\end{proof}

Compared with the number of calls to $O_{F,\tilde{\epsilon}}$ in the naive iteration method, which scales with $n$ as $O(n)$ and $O\left(n\times{\rm polylog}(n)\right)$ in the cases of smooth and nonsmooth $F$, respectively, that in the sum-in-QAE method more mildly scales as $O\left({\rm polylog}(n)\right)$.
This is because the sum-in-QAE method takes the sum in central difference formula by QAE and the normalization factor $D^{(m)}_n$ in the QAE target state $\ket{\Phi}$ is $O\left({\rm polylog}(n)\right)$.

Note that the sum-in-QAE method works for differentiation of expected values like $V$, or, more broadly, functions calculated by QAE, but not for general functions.
That is, if a function $f$ is defined as a summation like $V$ and QAE is used for the sum, we can `mix' the sum in the central difference formula into the sum in $f$, and simultaneously perform the two sums in one QAE.
On the other hand, if $f$ is calculated without the aid of QAE, we can calculate $\mathcal{D}_{n,m,h}[f](x)$ by naively iterating calculation and summation of $f(x+(-n)h),...,f(x+nh)$ faster than the way that they are computed in quantum parallel and summed up by QAE.
This is because the number of calls to $f$ in the naive iteration is at most $2n+1$, which is $O\left({\rm polylog}\left(\frac{1}{\epsilon}\right)\right)$ even in the setting $n=n_{\rm th}$ and $h=h_{\rm th}$, but that in the QAE-based calculation is $O\left(\frac{1}{\epsilon}\right)$.

Let us comment also on implementation of $O_{\rm coef}^{m,n}$ and $O_{\rm sign}^{m,n}$.
$O_{\rm coef}^{m,n}$ is the circuit to load the $2n+1$ precalculated numbers $\left|d^{(m)}_{n,-n}\right|,...,\left|d^{(m)}_{n,n}\right|$ as the quantum state $\ket{\Psi_{\rm coef}^{m,n}}$, and implemented using $\Theta(n)$ elementary gates \cite{Mottonen,Bergholm,Shende,Plesch,Iten,Park,Araujo}.
$O_{\rm sign}^{m,n}$ is the circuit to return 1 or 0 for each $j\in[-n:n]$ in the predetermined way, and implemented by $n$ multi-controlled NOT gates.
Thus, for $n$ which is at most logarithmically large as $n_{\rm th}$, we regard these oracles as less time-consuming than $O_{F,\tilde{\epsilon}}$.

\subsection{\label{sec:compare} Comparison of the methods}



\begin{table*}[t]
	\caption{The dependencies of the number of queries to $O_{F,\tilde{\epsilon}}$ on error tolerance $\epsilon$ in the proposed methods in different cases. The ones which can be better in each case are underlined.}
	\label{tbl:summary}
	\centering
	\begin{tabular}{wc{25mm}wc{30mm}||wc{15mm}wc{40mm}|wc{15mm}wc{40mm}}
		\hline
		\multirow{2}{*}{smoothness of $F$} & \multirow{2}{*}{\# of qubits available} & \multicolumn{2}{c|}{the naive iteration method} & \multicolumn{2}{c}{the sum-in-QAE method} \\
		& & $(n,h)$ & query number & $(n,h)$ & query number \\
		\hline \hline
		\multirow{2}{*}[3ex]{nonsmooth} \rule[0mm]{0mm}{8mm}& \multirow{2}{*}[3ex]{---} & \multirow{2}{*}[3ex]{$\left(n_{\rm th},h_{\rm th}\right)$} & \multirow{2}{*}[3ex]{$\frac{1}{\epsilon}\log^{m\sigma^+ + 1}\left(\frac{1}{\epsilon}\right)\log^m\left(\log\left(\frac{1}{\epsilon}\right)\right)$} & \multirow{2}{*}[3ex]{$\left(n_{\rm th},h_{\rm th}\right)$} & \multirow{2}{*}[3ex]{\underline{$\frac{1}{\epsilon}\log^{m\sigma^+}\left(\frac{1}{\epsilon}\right)\log^m\left(\log\left(\frac{1}{\epsilon}\right)\right)$}} \\
		\hline
		\multirow{2}{*}{smooth} & \multirow{2}{*}[3ex]{large} \rule[0mm]{0mm}{8mm} & \multirow{2}{*}[3ex]{$\left(\left\lceil\frac{m}{2}\right\rceil,h_{\rm min}\right)$} &  \multirow{2}{*}[3ex]{\underline{$\frac{1}{\epsilon}$}} & \multirow{2}{*}[3ex]{$\left(n_{\rm th},h_{\rm th}\right)$} & \multirow{2}{*}[3ex]{$\frac{1}{\epsilon}\log^{m\sigma^+}\left(\frac{1}{\epsilon}\right)\log^m\left(\log\left(\frac{1}{\epsilon}\right)\right)$} \\
		 & \multirow{2}{*}[3ex]{small} \rule[0mm]{0mm}{8mm}& \multirow{2}{*}[3ex]{$\left(n_{\rm th},h_{\rm th}\right)$} & \multirow{2}{*}[3ex]{\underline{$\frac{1}{\epsilon}\log\left(\frac{1}{\epsilon}\right)$}} & \multirow{2}{*}[3ex]{$\left(n_{\rm th},h_{\rm th}\right)$} & \multirow{2}{*}[3ex]{\underline{$\frac{1}{\epsilon}\log^{m\sigma^+}\left(\frac{1}{\epsilon}\right)\log^m\left(\log\left(\frac{1}{\epsilon}\right)\right)$}} \\
		\hline
	\end{tabular}
\end{table*}

Now, let us compare the presented methods in terms of the dependency of the number of queries to $O_{F,\tilde{\epsilon}}$ on error tolerance $\epsilon$, and discuss which one is better in each case.
Table \ref{tbl:summary} summarizes this, displaying the query numbers in various situations.
We consider the cases of smooth and nonsmooth $F$.
With respect to qubit capacity, we consider the situation where we can use as many qubits as we want and the opposite situation where the qubits available are limited and we want to save the qubit number.
Then, we consider the two settings, $n=\left\lceil\frac{m}{2}\right\rceil,h=h_{\rm min}$ and $n=n_{\rm th},h=h_{\rm th}$, and assume that the former is possible only in the large qubit capacity case.
When $F$ is nonsmooth, regardless of qubit capacity, the sum-in-QAE method is better than the naive iteration method by a factor $\log\left(\frac{1}{\epsilon}\right)$.
On the other hand, when $F$ is smooth and many qubits are available, the naive iteration method is better.
The discussion in the case of smooth $F$ and the small qubit capacity depends on comparison between factors $\log\left(\frac{1}{\epsilon}\right)$ and $\log^{m\sigma^+}\left(\frac{1}{\epsilon}\right)\log^m\left(\log\left(\frac{1}{\epsilon}\right)\right)$.
If $m\sigma^+<1$, the sum-in-QAE method can be better than the naive iteration method.
In particular, when $\sigma\le 0$, the logarithmic factor for the naive iteration method is replaced with the doubly logarithmic factor in the sum-in-QAE method, and therefore the latter method is promising.
Note that, of course, the above discussion is comparison of the asymptotic upper bounds of query numbers with constant factors omitted, and the actual best method can vary depending on the problem.



\section{Summary \label{sec:sum}}

In this paper, we considered the quantum methods to calculate derivatives of an expected value with a parameter, that is, $V(x)=E[F(S,x)]$, where $F$ is a function of a stochastic variable $S$ and a real parameter $x$ and $E[\cdot]$ denotes the expectation with respect to the randomness of $S$.
This is related to financial derivative pricing, an important industrial application of QMCI, since calculation of sensitivities of financial derivatives falls into this problem.
Since naively applying some finite difference formula to $V(x)$ leads to a poor accuracy due to the error in calculating $V$ divided by a small difference width $h$, we adopted the direction that we apply the central difference formula to $F$ and estimate its expected value.
Then, given some oracles such as $O_{F,\epsilon}$, which calculates $F$ with finite precision, and $O_S$, which generates a quantum state corresponding to a probability distribution $S$, we concretely presented two quantum methods, and evaluated their query complexities, focusing on the dependency on the error tolerance $\epsilon$.
The first method is the naive iteration method, in which we calculate the difference formula by simply iterating calls to $O_{F,\epsilon}$ for the terms in the formula, and then estimate the expected value by QAE.
The second one is the sum-in-QAE, in which we `mix' the summation of the terms into the sum over the possible values of $S$, and perform these sums in one QAE at one time.
We saw that there are some issues on the smoothness of $F$ and the number of qubits available.
First, if $F$ is nonsmooth with respect to $x$ and we yet take small $h$, the value of the difference formula on $F$ can be large for some $(S,x)$, and this leads to the large complexity in QAE.
Second, even if $F$ is smooth, in order to take small $h$, we need calculate $V$ with high precision, which means that we have to use many qubits for the calculation.
Considering these points, we saw that either of two methods can be advantageous against the other depending on the situation.
When $F$ is smooth and we can use many qubits, the naive iteration method with the lowest order difference formula and small $h$ is better.
Conversely, if $F$ is nonsmooth, we can use the higher order formula with $h$ that is $\widetilde{O}(1)$ with respect to $\epsilon$, and the sum-in-QAE method is better.
Even when $F$ is smooth, if we want to save qubits by using the higher order formula with $h=\widetilde{O}(1)$, the sum-in-QAE method is better depending on the parameter $\sigma$, which measures the smoothness of $F$.
In any case, we can calculate the derivative of $V$ with $\widetilde{O}\left(\frac{1}{\epsilon}\right)$ complexity, which is same as that for calculating $V$ itself, except for logarithmic factors.

We believe that the discussion in this paper provides us with insights on the plausible situation in the future, where we want to apply quantum algorithms to complicated industrial problems consuming many qubits, but the number of qubits available is limited.  
For future works, we will further aim to search industrial applications of quantum algorithms to concrete problems in a practical setting.

\section*{Acknowledgment}

This work was supported by MEXT Quantum Leap Flagship Program (MEXT Q-LEAP) Grant Number JPMXS0120319794.

\appendix

\section{Proof of Lemma \ref{lem:R} \label{sec:PrLemR}}

\begin{proof}
	First, we consider the case where $m\ge 2$.
	From the definition of $a^{(m)}_{n,j}$ in (\ref{eq:centGen}), 
	\begin{equation}
		|a^{(m)}_{n,j}| \le \sum_{\substack{\{l_1,...,l_{m-1}\}\in \qquad\quad \\ \mathcal{P}_{m-1}([-n:n]\setminus\{0,j\})}} \frac{(n!)^2}{|jl_1\cdots l_{m-1}|} \label{eq:a}
	\end{equation}
	holds for any $j\in[-n:n]\setminus\{0\}$.
	Note that
	\begin{eqnarray}
		&& \sum_{\substack{\{l_1,...,l_{m-1}\}\in \qquad\quad \\ \mathcal{P}_{m-1}([-n:n]\setminus\{0,j\})}} \frac{1}{|l_1\cdots l_{m-1}|} \nonumber \\
		& \le & \frac{1}{(m-1)!} \left(\frac{1}{|-n|}+\cdots+\frac{1}{|-1|}+\frac{1}{1}+\cdots+\frac{1}{n}\right)^{m-1} \nonumber \\
		&\le& \frac{\left(2n\right)^{m-1}}{(m-1)!} \nonumber\\
		&\le& \frac{m^{2n}}{(m-1)!}.
	\end{eqnarray}
	Here, the third inequality follows from $\left(2n\right)^{m-1}\le m^{2n}$, which holds for any positive integers $m$ and $n$ such that $m\le 2n$.
	Then, we see that
	\begin{equation}
		|a^{(m)}_{n,j}| \le \frac{(n!)^2m^{2n}}{|j|(m-1)!}. 
	\end{equation}
	This holds also when $m=1$, since
	\begin{equation}
	\left|a^{(1)}_{n,j}\right|=\left|\frac{\prod_{i\in[-n:n]\setminus\{0\}}i}{j}\right|=\frac{(n!)^2}{|j|}.
	\end{equation}
	Hence, combining this and (\ref{eq:f2n1Bound}) with the definition of $R_f(x,n,m,h)$ in (\ref{eq:centGen}), we have
	\begin{eqnarray}
		|R_f(x,n,m,h)| &\le& \frac{Mm!\left(n!\right)^2m^{2n}}{(2n+1)!(m-1)!}\sum_{j=-n}^n\frac{|j|^{2n}}{(n+j)!(n-j)!} \nonumber \\
		&\le& \frac{Mm\left(n!\right)^2n^{2n}m^{2n}}{(2n+1)!(2n)!}\sum_{j=-n}^n\binom{2n}{n+j} \nonumber \\
		&\le& \frac{Mm\left(n!\right)^2(2n)^{2n}m^{2n}}{(2n+1)!(2n)!} \nonumber \\
		&\le& \frac{Mm(2n)^{2n}m^{2n}}{(2n+1)}\left(\frac{n^{n+\frac{1}{2}}e^{-n}e^{\frac{1}{12n}}}{(2n)^{2n+\frac{1}{2}}e^{-2n}e^{\frac{1}{24n+1}}}\right)^2\nonumber \\
		&=& \frac{Mm}{2(2n+1)}\left(\frac{em}{2}\right)^{2n}e^{\frac{1}{6n}-\frac{2}{24n+1}}\nonumber \\
		&\le& Mm \left(\frac{em}{2}\right)^{2n}.
	\end{eqnarray}
	Here, we used $\sum_{j=0}^k\binom{k}{j}=2^k$, which holds for any $k\in\mathbb{N}$, at the third inequality, and
	\begin{equation}
		\sqrt{2\pi}n^{n+\frac{1}{2}}e^{-n}e^{\frac{1}{12n+1}}<n!<\sqrt{2\pi}n^{n+\frac{1}{2}}e^{-n}e^{\frac{1}{12n}}, \label{eq:Stirling}
	\end{equation}
	which is given in \cite{Robbins}, at the fourth inequality.
\end{proof}

\section{Proof of Lemma \ref{lem:cj} \label{sec:PrLemcj}}

\begin{proof}
	When $m\ge 2$, because of (\ref{eq:a}), we have
	\begin{eqnarray}
		\left|d^{(m)}_{n,j}\right|&\le&\frac{m!(n!)^2}{(n+j)!(n-j)!|j|}\sum_{\substack{\{l_1,...,l_{m-1}\}\in \qquad\quad \\ \mathcal{P}_{m-1}([-n:n]\setminus\{0,j\})}} \frac{1}{|l_1\cdots l_{m-1}|} \nonumber \\
		&\le& \frac{m!}{|j|}\sum_{\substack{\{l_1,...,l_{m-1}\}\in \qquad\quad \\ \mathcal{P}_{m-1}([-n:n]\setminus\{0,j\})}} \frac{1}{|l_1\cdots l_{m-1}|} \label{eq:temp}
	\end{eqnarray}
	for any $j\in[-n:n]\setminus\{0\}$, where we used
	\begin{equation} \frac{(n!)^2}{(n+j)!(n-j)!}=\frac{n(n-1)\cdots(n-j+1)}{(n+j)(n+j-1)\cdots(n+1)}\le1. \nonumber
	\end{equation}
	Therefore, we obtain
	\begin{eqnarray}
		D^{(m)}_n&=&\sum_{j=-n}^n\left|d^{(m)}_{n,j}\right| \nonumber \\
		&\le& 2\sum_{j\in[-n:n]\setminus\{0\}}\left|d^{(m)}_{n,j}\right|\nonumber \\
		&\le& \sum_{j\in[-n:n]\setminus\{0\}}\frac{2m!}{|j|}\sum_{\substack{\{l_1,...,l_{m-1}\}\in \qquad\quad \\ \mathcal{P}_{m-1}([-n:n]\setminus\{0,j\})}} \frac{1}{|l_1\cdots l_{m-1}|} \nonumber \\
		&\le& 2m!m \sum_{\substack{\{l_1,...,l_{m}\}\in \qquad\quad \\ \mathcal{P}_{m}([-n:n]\setminus\{0\})}} \frac{1}{|l_1\cdots l_{m}|}\nonumber \\
		&\le& 2m!m\times \frac{1}{m!} \left(\frac{1}{|-n|}+\cdots+\frac{1}{|-1|}+\frac{1}{1}+\cdots+\frac{1}{n}\right)^m\nonumber\\
		&\le&2m\left[2\left(1+\log n\right)\right]^m, \label{eq:temp6}
	\end{eqnarray}
	where the first equality follows from the definition of $d^{(m)}_{n,0}$ in (\ref{eq:centGen}).
	
	When $m=1$, $|a^{(1)}_{n,j}| = (n!)^2/|j|$ for $j\in[-n:n]\setminus\{0\}$, and therefore 
	\begin{equation}
		\left|d^{(1)}_{n,j}\right|=\frac{(n!)^2}{(n+j)!(n-j)!|j|} \le \frac{1}{|j|}.
	\end{equation}
	This leads to
	\begin{eqnarray}
		D^{(m)}_n&\le&2 \sum_{j\in[-n:n]\setminus\{0\}}\left|d^{(m)}_{n,j}\right|\nonumber \\ &\le&4\sum_{j=1}^n\frac{1}{j}\nonumber \\
		&\le& 4(1+\log n),
	\end{eqnarray}
	which indicates that (\ref{eq:CnUB}) holds also for $m=1$.
\end{proof}

\section{Proof of Lemma \ref{lem:cenAppOrd} \label{sec:PrLemCenAppOrd}}

To prove Lemma \ref{lem:cenAppOrd}, let us introduce some additional lemmas.

\begin{lemma}
	For any $a,x\in\mathbb{R}_+$ satisfying $a^2\le x$,
	\begin{equation}
		x\ge a\log x \label{eq:xalogx}
	\end{equation}
	holds. \label{lem:xlogx}
\end{lemma}

\begin{proof}
	Consider a function $f_a(x)=x-a\log x$ on $\mathbb{R}_+$.
	$f_a^\prime(x)=1-\frac{a}{x}$, and therefore it takes the minimum $a(1-\log a)$ at $x=a$.
	Then, we consider the following cases.\\
	
	\noindent (i) $a\le e$
	
	Since the minimum of $f_a(x)$ is larger than $0$,
	\begin{equation}
	\forall x\in \mathbb{R}_+, f_a(x)\ge 0, \label{eq:temp7}
	\end{equation}
	which means (\ref{eq:xalogx}) holds for any $x\in \mathbb{R}_+$.\\
	
	\noindent (i) $a> e$
	
	As a special case of (\ref{eq:temp7}), $f_2(y)=y-2\log y\ge 0$ holds for any $y\in \mathbb{R}_+$, which indicates that $f_a(a^2)=a(a-2\log a)\ge 0$.
	Besides, $f_a^\prime(x)>0$ for any $x$ larger than $a^2$.
	Combining these, we see that, for any $x$ larger than $a^2$, $f_a(x)\ge0$ holds, and so does (\ref{eq:xalogx}).
	
\end{proof}

\begin{lemma}
	For any $a\in\mathbb{R}_{\ge 0}$, $x\in\mathbb{R}_+$, and $\epsilon\in\mathbb{R}_+$ such that
	\begin{equation}
		\epsilon\le 
		\begin{cases}
		2^{a-\left(\frac{a}{\log 2}\right)^2} &; \ {\rm for} \ a \ge \log 2 \\
		2^{a-1} &; \ {\rm for} \ a < \log 2
		\end{cases}
		 \label{eq:epsCondGen}
	\end{equation}
	and
	\begin{equation}
		x\ge \tilde{x}_{a,\epsilon} := \log_2\left(\frac{2^a}{\epsilon}\right)+a\log_2\left(\log_2\left(\frac{2^a}{\epsilon}\right)\right), \label{eq:xCond}
	\end{equation}
	\begin{equation}
		\frac{x^a}{2^x} \le \epsilon \label{eq:xa_2x}
	\end{equation}
	holds.
	\label{lem:xa2x}
\end{lemma}

\begin{proof}
When $a=0$, (\ref{eq:xa_2x}) becomes $2^{-x}\le \epsilon$ and holds trivially, since (\ref{eq:xCond}) becomes $x\ge \log_2\left(\frac{1}{\epsilon}\right)$.
Therefore, we consider the case where $a>0$ below.
First, note that, under the condition (\ref{eq:epsCondGen}),
\begin{equation}
	\log_2 \left(\frac{2^a}{\epsilon}\right) \ge \left(\frac{a}{\log 2}\right)^2
\end{equation}
holds, and this means
\begin{equation}
	\log_2 \left(\frac{2^a}{\epsilon}\right) \ge \frac{a}{\log 2} \log \left( \log_2 \left(\frac{2^a}{\epsilon}\right)\right) = a \log_2\left( \log_2 \left(\frac{2^a}{\epsilon}\right)\right), \label{eq:temp3}
\end{equation}
because of Lemma \ref{lem:xlogx}.
Then, we see that
\begin{eqnarray}
	\frac{(\tilde{x}_{a,\epsilon})^a}{2^{\tilde{x}_{a,\epsilon}}}
	&=& \frac{\left[\log_2 \left(\frac{2^a}{\epsilon}\right) + a \log_2\left( \log_2 \left(\frac{2^a}{\epsilon}\right)\right) \right]^a}{\frac{2^a}{\epsilon} \left(\log_2\left(\frac{2^a}{\epsilon}\right)\right)^a}\nonumber \\
	&\le& \frac{2^a \left(\log_2 \left(\frac{2^a}{\epsilon}\right)\right)^a}{\frac{2^a}{\epsilon} \left(\log_2\left(\frac{2^a}{\epsilon}\right)\right)^a}\nonumber \\
	&=&\epsilon,
\end{eqnarray}
where we use the definition (\ref{eq:xCond}) of $\tilde{x}_{a,\epsilon}$ at the first equality, and (\ref{eq:temp3}) at the inequality.

On the other hand, defining $g(x):=\frac{x^a}{2^x}$, we obtain
\begin{equation}
	\frac{d}{dx}\log g(x) = \frac{a}{x}-\log 2,
\end{equation}
which implies that, when $x\ge \frac{a}{\log 2}$, $\log g(x)$ is monotonically deceasing, and so is $g(x)$.
Besides, under (\ref{eq:epsCondGen}), we can see that both
\begin{equation}
	\log_2 \left(\frac{2^a}{\epsilon}\right) \ge 1
\end{equation}
and
\begin{equation}
	\log_2 \left(\frac{2^a}{\epsilon}\right) \ge \frac{a}{\log 2}
\end{equation}
hold by simple algebra, and therefore
\begin{equation}
	\tilde{x}_{a,\epsilon}=\log_2\left(\frac{2^a}{\epsilon}\right)+a\log_2\left(\log_2\left(\frac{2^a}{\epsilon}\right)\right)
	\ge \frac{a}{\log 2}. 
\end{equation}
Hence, we obtain
\begin{equation}
	\frac{x^a}{2^x} \le \frac{(\tilde{x}_{a,\epsilon})^a}{2^{\tilde{x}_{a,\epsilon}}} \le \epsilon
\end{equation}
for $x\ge \tilde{x}_{a,\epsilon}$.

	
\end{proof}

Now, let us prove Lemma \ref{lem:cenAppOrd}.

\begin{proof}[Proof of Lemma \ref{lem:cenAppOrd}]

	Under (\ref{eq:epsCond}) and (\ref{eq:nth}), Lemma \ref{lem:xa2x} implies that
	\begin{equation}
		\frac{(2n+1)^{m\sigma^+}}{2^{2n+1}} \le \epsilon^\prime. \label{eq:temp5}
	\end{equation}
	Then, we see that
	\begin{eqnarray}
		&& \left|R_f(x,n,m,h)h^{2n-m+1}\right| \nonumber \\
		&\le& Ac^{2n+1}\left((2n+1)!\right)^{\sigma}m \left(\frac{em}{2}\right)^{2n}h^{2n-m+1}  \nonumber \\
		&\le& Ac^{2n+1}\left((2n+1)^{\sigma^+}\right)^{2n+1}m \left(\frac{em}{2}\right)^{2n}h^{2n-m+1} \nonumber \\
		&=& \frac{2A}{e}\left(\frac{ecm(2n+1)^{\sigma^+} h}{2}\right)^{2n+1}h^{-m} \nonumber \\
		&\le& \frac{2(ecm)^mA}{e}\frac{(2n+1)^{m\sigma^+}}{2^{2n+1}} \nonumber \\
		&\le& \frac{2(ecm)^mA}{e}\epsilon^\prime \nonumber \\
		&=& \epsilon.
	 	\label{eq:temp4}
	\end{eqnarray}
	Here, at the first inequality, we use (\ref{eq:RnUB}) with $M=Ac^{2n+1}\left((2n+1)!\right)^{\sigma}$, since $f\in\mathcal{G}_{A,C,\sigma}$.
	We also use $((2n+1)!)^\sigma\le((2n+1)^{\sigma^+})^{2n+1}$ at the second inequality, (\ref{eq:hcond}) at the third inequality, (\ref{eq:temp5}) at the fourth inequality, and (\ref{eq:epsprime}) at the last equality.
	Because of Theorem \ref{th:numDiff}, (\ref{eq:temp4}) indicates that (\ref{eq:diffErrUB}) holds.
	 
\end{proof}

\end{document}